\documentclass[runningheads]{llncs}

\title{Matching Patterns with Variables Under Simon's Congruence}

\author{Pamela Fleischmann\inst{1} \and Sungmin Kim\inst{3} \and Tore Koß\inst{2} \and Florin Manea\inst{2} \and Dirk Nowotka\inst{1} \and Stefan Siemer\inst{2} \and Max Wiedenhöft\inst{1}}

\authorrunning{P. Fleischmann et al.}

\institute{Department of Computer Science, Kiel University, Germany, \email{$\{$fpa,dn,maw$\}$@informatik.uni-kiel.de}
\and Department of Computer Science, University of Göttingen, Germany \email{$\{$tore.koss,florin.manea,stefan.siemer$\}$@cs.uni-goettingen.de}
\and Department of Computer Science, Yonsei University, Republic of Korea \email{rena\_rio@yonsei.ac.kr}}

\date{June 2023}

\usepackage{amssymb,amsmath}
\usepackage{graphicx} 
\usepackage{algorithm} 
\usepackage{algpseudocode} 
\usepackage{complexity} 
\usepackage{comment} 
\usepackage[colorinlistoftodos,prependcaption,textsize=tiny]{todonotes} 
\usepackage{tabularx}
\usepackage{environ}
\usepackage{todonotes}
\usepackage{tikz}
\usetikzlibrary{backgrounds}
\usepgflibrary{arrows}
\usepackage{float}
\usepackage{cite}
\usepackage{mathtools}

\usepackage{thmtools,thm-restate}

\newcommand{\alphabetOf}[1]{\texttt{\upshape alph}(#1)}
\newcommand{\universalityIdx}[2]{\iota_{#1}(#2)}
\newcommand{\Xrank}[3]{\texttt{\upshape X}(#1,#2,#3)}

\newcommand{\subseq}[2]{#1\preceq#2}
\newcommand{\subseqset}[2]{\mathbb{S}_{#1}(#2)}
\newcommand{\rest}[2]{\texttt{\upshape rest}_{#1}(#2)}
\newcommand{\arch}[2]{\texttt{\upshape arch}_{#1}(#2)}
\newcommand{\archEnd}[3]{\texttt{\upshape archEnd}_{#1,#2}(#3)}
\newcommand{\marginalSeq}[2]{\texttt{\upshape M}_{#1}(#2)}
\newcommand{\subseqUnivSign}[1]{\texttt{\upshape s}(#1)}
\newcommand{\myKarr}[1]{\mathcal{K}}
\newcommand{\myRarr}[1]{\mathcal{R}}

\newcommand{\alphRest}[1]{\myRarr{}[#1]}

\newcommand{\set}[1]{\lbrace #1 \rbrace}

\newcommand{\match}{\mathtt{Match}}
\newcommand{\simmatch}{\mathtt{MatchSimon}}
\newcommand{\simStrictMatch}{\mathtt{MatchStrictSimon}}
\newcommand{\univmatch}{\mathtt{MatchUniv}}
\newcommand{\simWE}{\mathtt{WESimon}}
\newcommand{\simStrictWE}{\mathtt{WEStrictSimon}}

\newcommand{\sat}{\mathtt{3CNFSAT}}

\newcommand{\var}{\mathtt{var}}
\newcommand{\term}{\mathtt{term}}

\def\nth#1{#1$^{\tiny th}$}

\def\N{\mathbb{N}}
\def\ta{\mathtt{a}}

\DeclareMathOperator*{\dollar}{\texttt{\$}}
\DeclareMathOperator*{\hashtag}{\texttt{\#}}
\DeclareMathOperator{\emptyword}{\varepsilon}

\makeatletter
\newcommand{\problemtitle}[1]{\gdef\@problemtitle{#1}}
\newcommand{\probleminput}[1]{\gdef\@probleminput{#1}}
\newcommand{\problemquestion}[1]{\gdef\@problemquestion{#1}}
\NewEnviron{problemdescription}{
  \problemtitle{}\probleminput{}\problemquestion{}
  \BODY
  \par\addvspace{.5\baselineskip}
  \noindent
  \begin{tabularx}{\textwidth}{@{\hspace{\parindent}} l X c}
    \hline
    \multicolumn{2}{@{\hspace{\parindent}}l}{\@problemtitle} \\
    \textbf{Input:} & \@probleminput \\
    \textbf{Question:} & \@problemquestion\\
    \hline
  \end{tabularx}
  \par\addvspace{.5\baselineskip}
}
\makeatother

\usetikzlibrary{positioning, chains, calc, arrows, decorations.pathreplacing, fit, patterns, graphs, external, calligraphy}
\usetikzlibrary{automata,positioning}
\usetikzlibrary{shapes,backgrounds}
\usetikzlibrary{matrix}
\usetikzlibrary{chains}

\tikzset{ns/.style={
        inner sep = 0mm, outer sep=0mm,
        append after command={
            [very thick]
            (\tikzlastnode.north west)edge(\tikzlastnode.north east)
            [very thick]
            (\tikzlastnode.south west)edge(\tikzlastnode.south east)
        }
    }
}

\tikzset{wens/.style={
        inner sep = 0mm, outer sep=0mm,
        append after command={
            [very thick]
            (\tikzlastnode.north west)edge(\tikzlastnode.south west)
            [very thick]
            (\tikzlastnode.north east)edge(\tikzlastnode.south east)
            [very thick]
            (\tikzlastnode.north west)edge(\tikzlastnode.north east)
            [very thick]
            (\tikzlastnode.south west)edge(\tikzlastnode.south east)
        }
    }
}

\tikzset{wns/.style={
        inner sep = 0mm, outer sep=0mm,
        append after command={
            [very thick]
            (\tikzlastnode.north west)edge(\tikzlastnode.south west)
            [very thick]
            (\tikzlastnode.north west)edge(\tikzlastnode.north east)
            [very thick]
            (\tikzlastnode.south west)edge(\tikzlastnode.south east)
        }
    }
}

\tikzset{ens/.style={
        inner sep = 0mm, outer sep=0mm,
        append after command={
            [very thick]
            (\tikzlastnode.north east)edge(\tikzlastnode.south east)
            [very thick]
            (\tikzlastnode.north west)edge(\tikzlastnode.north east)
            [very thick]
            (\tikzlastnode.south west)edge(\tikzlastnode.south east)
        }
    }
}

\tikzset{e/.style={
        inner sep = 0mm, outer sep=0mm,
        append after command={
            [very thick]
            (\tikzlastnode.north east)edge(\tikzlastnode.south east)
            [very thick]
        }
    }
}

\tikzset{w/.style={
        inner sep = 0mm, outer sep=0mm,
        append after command={
            [very thick]
            (\tikzlastnode.north west)edge(\tikzlastnode.south west)
            [very thick]
        }
    }
}	

\begin{document}

\maketitle

\begin{abstract}
We introduce and investigate a series of matching problems for patterns with variables under Simon's congruence. Our results provide a thorough picture of these problems' computational complexity. \end{abstract}


\section{Introduction}\label{section:introduction}

A \emph{pattern with variables} is a string $\alpha$ consisting of \emph{constant letters} (or {\em terminals}) from a finite alphabet $\Sigma=\{1,\ldots,\sigma\}$, of size $\sigma\geq 2$, and \emph{variables} from a potentially infinite set $\mathcal{X}$, with $\Sigma\cap \mathcal{X} = \emptyset$. Such a pattern $\alpha$ is mapped by a function $h$, called {\em substitution}, to a word by substituting the variables occurring in $\alpha$ by arbitrary strings of constants, i.e., strings over $\Sigma $. For example, the pattern $\alpha=x x \mathtt{abab} y y $ can be mapped to the string of constants $\mathtt{aaaa abab bb}$ by the substitution $h$ defined by $h(x)=\mathtt{aa}, h(y)=\mathtt{b}$. In this framework, $h(\alpha)$ denotes the word obtained by substituting every occurrence of a variable $x$ in $\alpha$ by $h(x)$ and leaving all the constants unchanged. If a pattern $\alpha$ can be mapped to a string of constants $w$, we say that $\alpha$ matches $w$;
the problem of deciding, given a pattern $\alpha$ with variables and a string of constants $w$, whether there exists a substitution which maps $\alpha$ to $w$ is called the {\em (exact) matching problem}, $\match$. \looseness=-1

\begin{problemdescription}
  \problemtitle{Exact Matching Problem: $\match(\alpha,w)$}
  \probleminput{Pattern $\alpha$, $|\alpha|=m$, word $w$, $|w|=n$.}
  \problemquestion{Is there a substitution $h$ with $h(\alpha) = w$?}
\end{problemdescription}

$\match$ is a heavily studied problem, which appears frequently in various areas of theoretical computer science. Initially, this problem was considered in language theory (e.g., pattern languages~\cite{DBLP:journals/jcss/Angluin80}) or combinatorics on words (e.g., unavoidable patterns~\cite{Loth02}), with connections to algorithmic learning theory (e.g., the theory of descriptive patterns for finite sets of words~\cite{DBLP:journals/jcss/Angluin80,shi:pat,DBLP:journals/tcs/FernauMMS18}), and has by now found interesting applications in string solving and the theory of word equations (\cite{lothaire}), stringology (e.g., generalised function matching~\cite{ami:gen}), the theory of extended regular expressions with backreferences~\cite{cam:afo,fri:mas,Fre2013,FreydenbergerSchmid2019}), or database theory (mainly in relation to document spanners~\cite{FreydenbergerHolldack2018,Freydenberger2019,FaginEtAl2015,SchmidSchweikardt2021,FreydenbergerP21,SchmidICDT2022,SchmidSchweikardtPODS2022}).\looseness=-1

$\match$ is \NP-complete in general \cite{DBLP:journals/jcss/Angluin80}, and a more detailed image of the parameterised complexity of the matching problem is revealed in~\cite{ReidenbachS14,shi:pol2,FerSch2015,DBLP:journals/mst/FernauSV16,DBLP:journals/toct/FernauMMS20,schmid13} and the references therein. A series of classes of patterns, defined by structural restrictions, for which $\match$ is in \P~were identified~\cite{ReidenbachS14,DayFMN17,DBLP:journals/toct/FernauMMS20}; moreover, for most of these classes, $\match$ is $W[1]$-hard \cite{w1hardness} with respect to the structural parameters used to define the respective classes.  
Recently, Gawrychowski et al.~\cite{mfcs2021,spire2022} studied $\match$ in an approximate setting: given a pattern $\alpha$, a word $w$, and a natural number $\ell$, one has to decide if there exists a substitution $h$ such that $D(h(\alpha),w)\leq \ell$, where $D$ is either the Hamming  \cite{mfcs2021} or the edit distance \cite{spire2022}. Their results offered, once more, a detailed understanding of the approached matching problems' complexity (in general, and for classes of patterns defined by structural restrictions). \looseness=-1
The problems discussed in~\cite{mfcs2021,spire2022} can be seen in a more general setting: given a pattern $\alpha$ and a word $w$, decide if there exists a substitution $h$ such that $h(\alpha)$ is similar to $w$, with respect to some similarity measure (Hamming resp. edit distance in ~\cite{mfcs2021,spire2022} or string equality for exact $\match$).
Thus, it seems natural to also consider various other string-equivalence relations as similarity measures, such as ($k$-)abelian equivalence \cite{KarhumakiSZ13,KarhumakiSZ17} or $k$-binomial equivalence\cite{RigoS15,FreydenbergerGK15,LejeuneRR20}. Here, we consider an approximate variant of $\match$ using Simon's congruence $\sim_k$ \cite{Simon75}.


\begin{problemdescription}
  \problemtitle{Matching under Simon's Congruence: $\simmatch(\alpha,w,k)$}
  \probleminput{Pattern $\alpha$, $|\alpha|=m$, word $w$, $|w|=n$, and number $k\in [n]$.}
  \problemquestion{Is there a substitution $h$ with $h(\alpha) \sim_k w$?}
\end{problemdescription}

Let us recall the definition of Simon's congruence. A string~$u$ is a \emph{subsequence} of a  string~$v$ if $u$ results from $v$ by deleting some letters of $v$. Subsequences are well studied in the area of combinatorics of words and combinatorial pattern matching, and are well-connected to other areas of computer science (e.g., the handbook \cite{lothaire} or the survey \cite{Kosche2022SubsequenceSurvey} and the references therein). Let $\subseqset{k}{v}$ be the set of all  subsequences of a given string $v$ up to length $k\in\N_0.$ Two strings~$v$ and $v'$ are $k$-Simon congruent iff $\subseqset{k}{v} = \subseqset{k}{v'}.$ The problem of testing whether two given strings are $k$-Simon congruent, for a given $k$, was introduced by Imre Simon in his PhD thesis \cite{simonPhD} as a similarity measure for strings, and was intensely studied in the combinatorial pattern matching community (see \cite{Hebrard91,garelCPM,SimonWords,DBLP:conf/wia/Tronicek02,CrochemoreMT03,FleischerK18} and the references therein), before being optimally solved in \cite{BarkerFHMN20,GawrychowskiKKM21}. Another interesting extension of these results, discussed in \cite{patternMatchingSimon}, brings us closer to the focus of this paper. There, the authors  present an efficient solution for the following problem: given two words $w,u$ and a natural number $k$, decide whether there exists a factor of $w$ which is $k$-Simon congruent to $u$; this is $\simmatch$ with the input pattern $\alpha=x u y$ for variables $x,y$. Thus, it seems natural to consider, in a general setting, the problem of checking whether one can map a given pattern $\alpha$ to a string which is similar to $w$ with respect to $\sim_k$. Moreover, there is another way to look at this problem, which seems interesting to us: the input word $w$ and the number $k$ are a succinct representation of  $\subseqset{k}{w}$. So, $\simmatch(\alpha,w,k)$ asks whether or not we can assign the variables of $\alpha$ in such a way that we reach a word describing the target set of subsequences of length $k$, as well. \looseness=-1

One of the congurence-classes of $\Sigma^*$ w.r.t. $\sim_k$ received a lot of attention: the class of $k$-subsequence universal words, those words which contain all $k$-length words as subsequences. This class was first studied in \cite{karandikar2016height,schnoebelen2019height}, and further investigated in~\cite{BarkerFHMN20,day2021edit,fleischmann2021scattered,kosche2021absent,adamson2023words,Goettingen2023words,SchnoebelenV23,fleischmann2023alphabetafactorization} in contexts related to and motivated by formal languages, automata theory, or combinatorics, where the notion of universality is central (see   \cite{Rampersad:2012,KrotzschMT17,GawrychowskiRSS17,martin1934,Bruijn46,ChenKMS17,GoecknerGHKKKS18} for examples in this direction). The motivation of studying $k$-subsequence universal words is thoroughly discussed in \cite{day2021edit}. Here, we consider the following problem: \looseness=-1

\begin{problemdescription}
  \problemtitle{Matching a Target Universality: $\univmatch(\alpha,k)$}
  \probleminput{Pattern $\alpha$, $|\alpha|=m$, and $k\in \N_0$.}
  \problemquestion{Is there a substitution $h$ with $\iota(h(\alpha))=k$?}
\end{problemdescription}
In this problem, $\iota(w)$ (the universality index of $w$) is the largest integer $\ell$ for which $w$ is $\ell$-subsequence universal. Note that $\univmatch$ can be formulated in terms of $\simmatch$: the answer to $\univmatch(\alpha,k)$ is yes if and only if the answer to $\simmatch(\alpha, (1\cdots \sigma)^k, k)$ is yes and the answer to $\simmatch(\alpha$, $(1\cdots \sigma)^{k+1}, k+1)$ is no. However, there is an important difference: for $\univmatch$ we are not explicitly given the target word $w$, whose set of $k$-length subsequences we want to reach; instead, we are given the number $k$ which represents the target set more compactly (using only $\log k$ bits). \looseness=-1

In the problems introduced above, we attempt to match (or reach), starting with a pattern $\alpha$, the set of subsequences defined by a given word $w$ (given explicitly or implicitly). A well-studied extension of $\match$ is the satisfiability problem for word equations, where we are given two patterns $\alpha$ and $\beta$ and are interested in finding an assignment of the variables that maps both patterns to the same word (see, e.g., \cite{lothaire}). This problem is central both to combinatorics on words and to the applied area of string solving \cite{Amadini2020,Hague19}. In this paper, we extend $\simmatch$ to the problem of solving word equations under $\sim_k$, defined as follows.
\begin{problemdescription}
  \problemtitle{Word Equations under Simon's Congruence: $\simWE(\alpha,\beta,k)$}
  \probleminput{Patterns $\alpha$, $\beta$, $|\alpha|=m$, $|\beta|=n$, and $k\in [m+n]$.}
  \problemquestion{Is there a substitution $h$ with $h(\alpha) \sim_k h(\beta)$?}
\end{problemdescription}

Besides introducing these natural problems, our paper presents a rather comprehensive picture of their computational complexity. We start with $\univmatch$, the most particular of them and whose input is given in the most compact way. In Section \ref{section:patternMatchingUniversality} we show that $\univmatch$ is \NP-complete, and also present a series of structurally restricted classes of patterns, for which it can be solved in polynomial time. In Section \ref{section:patternMatchingSimk}, we approach $\simmatch$ and show that it is also \NP-complete; some other variants of this problem, both tractable and intractable, are also discussed. Finally, in Section \ref{section:wordEqSimk}, we discuss $\simWE$ and its variants, and characterise their computational complexity. The paper ends with a section pointing to a series of future research directions.

\section{Preliminaries} \label{section:preliminaries}

Let $\mathbb{N} = \{1, 2, \ldots\}$ be the set of natural numbers. Let $[n] = \{1, \ldots, n\}$ and 
$[m:n]=[n]\setminus [m-1]$, for $m,n \in \mathbb{N}, m < n$. $\N_0$ denotes $\mathbb{N}\cup\{0\}$.\looseness=-1

For a finite set $\Sigma = [\sigma]$ called {\em alphabet}, $\Sigma^*$ denotes the set of all words (or strings) over $\Sigma$, with $\varepsilon$ denoting the empty word. For $w \in \Sigma^*$, $|w|$ denotes its length, while $|w|_\mathtt{a}$ denotes the number of occurrences of $\mathtt{a}\in \Sigma$ in $w$. Further, $\Sigma^{\leq k}$ (resp. $\Sigma^k$) denotes the set of all words over $\Sigma$ up to (resp. of) length $k\in\mathbb{N}$. 
Let $w[i]$ denote the \nth{$i$} letter in the string~$w$, and let $\alphabetOf{w}=\{\mathtt{a}\mid |w|_{\mathtt{a}}\geq 1 \}$ denote the set of different letters in $w$. 
To access the first occurrence of a letter $\ta\in\Sigma$ after a position $i\in[|w|]$ in a  word $w\in\Sigma^{\ast}$, define the {\em X-ranker} as a mapping $\mathtt{X}:\Sigma^{\ast}\times([|w|]\cup\{0,\infty\})\times \Sigma\rightarrow [|w|]\cup\{\infty\}$ with $(w,i,\mathtt{a})\mapsto
\min(\{j\in[i+1:|w|]\mid w[j]=\mathtt{a}\}\cup\{\infty\})$ (cf.\cite{DBLP:journals/corr/abs-0907-0616}). Notice that a lookup table
for all possible X-ranker evaluations for some given $w\in\Sigma^{\ast}$ can be computed in linear time in $|w|$, where each item can be accessed in constant time~\cite{FleischerK18,BarkerFHMN20}. In the special case of $\Xrank{w}{0}{\ta}$, we call this occurrence of $\ta$ the {\em signature letter} $\ta$ of $w$, for all $\ta\in\alphabetOf{w}$. 
A \emph{permutation}~$\gamma$ of an alphabet~$\Sigma$
is a string in $\Sigma^\sigma$
with $\alphabetOf{\gamma}=\Sigma$.
A string~$u$ is a \emph{subsequence} of a string~$w$
if there exists a strictly increasing
integer sequence~$0<i_1<i_2<\ldots <i_{|u|}\le |w|$
with $w[i_j]=u[j]$ for all $j\in[|u|]$.
For a given $k\in\mathbb{N}_0$, we use $\subseqset{k}{w}$  as
the set of all subsequences of $w$ with length at most $k$.
A subsequence $u$ of $w$ is called a {\em substring} of $w$ if there exists a position $i$ of $w$ such that $u=w[i]w[i+1]\cdots w[i+|u|-1]$. We write $w[i:j]$ for $w[i]w[i+1]\cdots w[j]$ for $1\leq i\leq j\leq|w|$.
Substrings $w[1:j]$ (resp., $w[i:|w|]$) are called \emph{prefixes} (resp., \emph{suffixes}) of $w$.

Two words $w_1,w_2\in\Sigma^{\ast}$ are called {\em Simon $k$-congruent} ($w_1\sim_k w_2$) if $\subseqset{k}{w_1}=\subseqset{k}{w_2}$ {\upshape \cite{Simon75}}. A word $w\in\Sigma^{\ast}$ is called {\em $k$-subsequence universal} (or $k$-universal for short) for some $k\in\N$ if $\subseqset{k}{w}=\Sigma^{\leq k}$; this means that $w\sim_k (1\cdots \sigma)^k $. The largest $k\in\N_0$ such that $w$ is $k$-universal is the {\em universality index} of $w$, denoted by $\iota(w)$. In \cite{Hebrard91}, Hébrard introduced the following unique factorisation of words.\looseness=-1

\begin{definition}
The {\em arch factorisation} of a word $w\in\Sigma^{\ast}$ is defined by $w=\arch{1}{w}\cdots\arch{k}{w}\rest{}{w}$ for some $k\in\N_0$ such that there exists a sequence $(i_j)_{j\leq k}$ with $i_0=0$,
    $i_j=\max\{\Xrank{w}{i_{j-1}}{\mathtt{a}}\mid \mathtt{a}\in\Sigma\}$ for all $j\ge 1$,
    $\arch{j}{w}=w[i_{j-1}+1:i_j]$ whenever $1\le i_j < \infty$, and $\rest{}{w}=w[i_j:|w|]$, if $i_{j+1}=\infty$.
\end{definition}

Clearly, the number of arches of $w\in\Sigma^{\ast}$ is exactly $\iota(w)$.
Extending the notion of arch factorisation,
we define the arches and rest of $w\in\Sigma^{\ast}$
for~$\mathtt{a}\in \alphabetOf{w}$
(cf. the arch jumping functions introduced in~\cite{SchnoebelenV23}) as well as the universality index
for the respective letter~$\mathtt{a}$.
That is, we perform the arch factorisation
and obtain the universality index
for the suffix of $w$
that starts after the first occurrence
of~$\mathtt{a}$.\looseness=-1

\begin{definition}
Let $w\in\Sigma^{\ast}$, $\ta\in\alphabetOf{w}$, and $j\in[\iota(w)]$. The arches of signature letters are defined by $\arch{\mathtt{a},j}{w}=\arch{j}{w[\Xrank{w}{0}{\mathtt{a}}+1:|w|]}$ and $\rest{\mathtt{a}}{w}=\rest{}{w[\Xrank{w}{0}{\mathtt{a}}+1:|w|]}$.
The {\em universality index of $\ta$} is $\universalityIdx{\mathtt{a}}{w}=\universalityIdx{}{w[X(w,0,\mathtt{a})+1:|w|]}$. The last index with respect to $w$ of $\arch{\mathtt{a},j}{w}$ is defined as $\archEnd{\mathtt{a}}{j}{w}=\Xrank{w}{0}{\ta} + \sum_{i=1}^j|\arch{\mathtt{a},i}{w}|$.
\end{definition}

Now, we are interested in the smallest substrings of $w$ that allow the completion of rests of specific prefixes of $w$ to full arches.
Hence, we define {\em marginal sequences}, which are breadth-first orderings 
of $\sigma$ parallel arch factorisations,
each starting after a signature letter of the word.

\begin{figure}
    \centering
\includegraphics[scale=1]{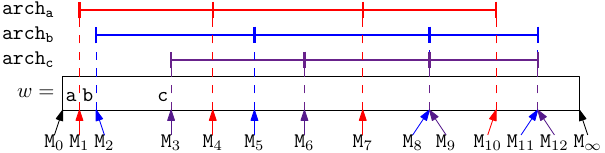}
    \caption{The marginal sequence of a word.}
    \label{fig:margSeq}
\end{figure}

\begin{restatable}[]{definition}{refmarginalSeq}\label{def:marginalSeq}
Let $w\in\Sigma^{\ast}$ and $\gamma$ be a permutation of $\Sigma$ such that $\Xrank{w}{0}{\gamma[i]}$
is increasing w.r.t. $i\in[\sigma]$.
From the arches for signature letters,
we define the \emph{marginal sequence} of integers
of $w\in\Sigma^{\ast}$ inductively by $\marginalSeq{0}{w}=0$, $\marginalSeq{i}{w}=\Xrank{w}{0}{\gamma[i]}$ for all $i\in[\sigma]$, and $\marginalSeq{i\sigma+j}{w}=\archEnd{\gamma[j]}{i}{w}$
    for $j\in[\sigma]$, $i\in[\universalityIdx{\gamma[j]}{w}]$. Let $\marginalSeq{\infty}{w}=|w|$
    denote the last element of the sequence. 
\end{restatable}

The sequence is called {\em marginal} because, for $j\in[\sigma]$, $w[\marginalSeq{i\sigma+j-1}{w}+1:\marginalSeq{i\sigma+j}{w}]$ is the smallest prefix $p$ of $w[\marginalSeq{i\sigma+j-1}{w}+1:|w|]$ such that $\universalityIdx{\gamma[j]}{w[1:\marginalSeq{i\sigma+j-1}{w}]p }=i$. 
    Note that the marginal sequence~$\marginalSeq{i}{w}$
    is non-decreasing.
In the following, we define a slight variation of the {\em subsequence universality signature} $\subseqUnivSign{w}$ introduced in~\cite{SchnoebelenV23}.

\begin{definition}
1. For $w\in\Sigma^{\ast}$, the \emph{subsequence universality signature} $\subseqUnivSign{w}$
of $w$ is defined as the $3$-tuple~$(\gamma,\myKarr{},\myRarr{})$ with a permutation $\gamma$ of $\alphabetOf{w}$,
    where $\Xrank{w}{0}{\gamma[i]}$ $>\Xrank{w}{0}{\gamma[j]} \Leftrightarrow i > j$ ($\gamma$ consists of the letters of $\alphabetOf{w}$ in order of their first appearance in $w$) and two arrays $\myKarr{}$ and $\myRarr{}$ of length $\sigma$ with $\myKarr{}[i]=\universalityIdx{\gamma[i]}{w}$ and 
    $\alphRest{i}=\alphabetOf{\rest{\gamma[i]}{w}}$ for all $i\in[|\alphabetOf{w}|]$.
    For all $i\in[\sigma]\setminus \alphabetOf{w}$,
    we have $\myRarr{}[i]=\Sigma$
    and $\myKarr{}[i]=-\infty$. \\
2. Conversely,
for a permutation~$\gamma'$ of $\Sigma$,
an integer array~$\myKarr{}'$
and an alphabet array~$\myRarr{}'$ both of length $\sigma$,
we say that the tuple~$(\gamma',\myKarr{}',\myRarr{}')$ is a \emph{valid} signature 
if there exists a string~$w$
that satisfies $\subseqUnivSign{w}=(\gamma,'\myKarr{}',\myRarr{}')$.    
\end{definition}
Note that, for $k_i=\universalityIdx{\gamma[i]}{w}$,
we have $\myRarr{}[i]
=\alphabetOf{w[\marginalSeq{k_i\sigma + i}{w}+1:
\marginalSeq{\infty}{w}]}$,
since $\rest{\gamma[i]}{w}=
w[\marginalSeq{k_i\sigma + i}{w}+1:
\marginalSeq{\infty}{w}]$.

A central notion to this work is that of patterns with variables. From now on, we consider two alphabets: $\Sigma = [\sigma]$ is an alphabet of constants (or terminals), and ${\mathcal X}$ a (possibly infinite) alphabet of variables, with ${\mathcal X}\cap \Sigma =\emptyset$. A pattern $\alpha$ is a string from $({\mathcal X} \cup \Sigma)^*$, i.e., a string containing both constants and variables. For a pattern $\alpha$, $\var(\alpha)=\alphabetOf{\alpha}\cap {\mathcal X}$ denotes the set of {\em variables} in $\alpha$, while $\term(\alpha)=\alphabetOf{\alpha}\cap \Sigma$ is the set of {\em constants} ({\em terminals}) in $\alpha$.

\begin{definition}
A \emph{substitution}~$h:({\mathcal X}\cup \Sigma)^\ast\to\Sigma^\ast$ is a morphism that acts as the identity on $\Sigma$ and maps each variable of ${\mathcal X}$ to a (potentially empty) string over $\Sigma$. That is, $h(\ta)=\ta$ for all $\ta\in\Sigma$ and $h(x)\in\Sigma^{\ast}$ for all $x\in\mathcal{X}$. We say that pattern~$\alpha$ \emph{matches} string~$w$ over $\Sigma$ under a binary relation~$\sim$ if there exists a substitution~$h$ that satisfies $h(\alpha)\sim w$. 
\end{definition}
In the above definition, if $\sim$ is the string equality $=$, we say that the pattern~$\alpha$ \emph{matches} the string~$w$ instead of saying that $\alpha$ \emph{matches} $w$ under $=$. 

The problems addressed in this paper, introduced in Section \ref{section:introduction}, deal with  matching patterns to words under Simon's congruence~$\sim_k$. For these problems, the input consists of patterns, words, and a number $k$. In general, we assume that each letter of $\Sigma$ appears at least once, in at least one of the input patterns or words. Therefore, for input pattern $\alpha$ and word $w$ we assume that $\Sigma=\term(\alpha)\cup\alphabetOf{w}$. Hence, $\sigma$ is upper bounded by the total length of the input words and patterns. Similarly, the total number of variables occurring in the input patterns is upper bounded by the total length of these patterns. However, in this paper, although the number of variables is not restricted, we assume that $\sigma$ is a constant, i.e., $\sigma\in O(1)$. Clearly, the complexity lower bounds proven in this setting for the analysed problems are stronger while the upper bounds are weaker than in the general case, when no restriction is placed on $\sigma$. Note, however, that $\sigma\in O(1)$ is not an unusual assumption, being used in, e.g., \cite{FleischerK18}.

\medskip

The \textbf{computational model} we use is the Word RAM model with memory words of logarithmic size. This is a standard computational model in algorithm design in which, for an input of size $n$, the memory consists of memory-words consisting of $\Theta(\log n)$ bits. Basic operations (including arithmetic and bitwise Boolean operations) on memory-words take constant time, and any memory-word can be accessed in constant time. Numbers larger than $n$, with $\ell$ bits, are represented in $\Theta(\ell/\log n)$ memory words, and working with them takes time proportional to the number of memory words on which they are represented. In all the problems, we assume that we are given a pattern $\alpha$, with $|\alpha|=n$, over a constant size alphabet of constants $\Sigma=\{1,2,\ldots,\sigma\}$, with $\sigma\in O(1)$, and a set of variables $X \coloneqq \{x_1,\ldots, x_n\}$ that can be encoded as integers between $1$ and $\sigma + n$. That is, we assume that the processed patterns are sequences of integers (called letters or symbols), each fitting in $O(1)$ memory words. This is a common assumption in string algorithms: the input is said to be over {\em an integer alphabet}. For instance, the same assumption was also used for developing efficient algorithms for $\match$ in \cite{DBLP:journals/tcs/FernauMMS18,mfcs2021}. 
For a more detailed general discussion on this computational model see, e.g.,~\cite{crochemore}.

\section{$\univmatch$}\label{section:patternMatchingUniversality}

In this section, we discuss the $\univmatch$ problem. In this problem, we are given a pattern $\alpha$ and a natural number $k\leq n$, and we want to check the existence of a substitution $h$ with $\iota(h(\alpha))=k$. Note that $\iota(h(\alpha))=k$ means both that $h(\alpha)$ is $k$-universal and that it is not $(k+1)$-universal. A slightly relaxed version of the problem, where we would only ask for $h(\alpha)$ to be $k$-universal is trivial (and, therefore, not interesting): the answer, in that case, is always positive, as it is enough to map one of the variables of $\alpha$ to $(1\cdots \sigma)^k$. The main result of this section is that $\univmatch$ is \NP-complete, which we will show in the following.

To show that $\univmatch (\alpha, k)$ is \NP-hard, we reduce $\sat$ (3-satisfiability in conjunctive normal form) to $\univmatch (\alpha,k)$. 
We provide several gadgets allowing us to encode a $\sat$-instance $\varphi$ as an $\univmatch$-instance $(\alpha,k)$. Finally, we show that we can find a substitution $h$ for the instance $(\alpha,k)$, such that $\iota(h(\alpha)) = k$, if and only if $\varphi$ is satisfiable. We begin by recalling $\sat$.  
\begin{problemdescription}
  \problemtitle{3-Satisfiability for formulas in conjunctive normal form, $\sat$.}
  \probleminput{Clauses $\varphi \coloneqq \set{c_1,c_2,\ldots,c_m}$, where $c_j = (y_j^1 \lor y_j^2 \lor y_j^3)$ for $1\leq j \leq m$, and $y_j^1,y_j^2,y_j^3$ from a finite set of boolean variables $X \coloneqq \set{x_1,x_2, \ldots, x_n}$ and their negations $\bar{X} \coloneqq \set{\bar{x}_1,\bar{x}_2,\ldots,\bar{x}_n}$.}
  \problemquestion{Is there an assignment for $X$, which satisfies all clauses of $\varphi$?}
\end{problemdescription}

It is well-known that $\sat$ is \NP-complete (see \cite{Karp72,gar:com} for a proof). With this result at hand, we can prove the following lower bound. 

\begin{restatable}[]{lemma}{refmatchUnivnphard}\label{lemma:matchUnivnphard}
    $\univmatch$ is \NP-hard.
\end{restatable}

\begin{proof}
We reduce $\sat$ to $\univmatch (\alpha,k)$. Let us consider an instance of $\sat$: formula $\varphi$ given by $m$ clauses $\varphi \coloneqq \set{c_1,c_2,\ldots c_m}$ over $n$ variables $X \coloneqq \set{x_1,x_2, \ldots x_n}$ (for simplicity in notation we define $N = n+m$). We map this $\sat$ instance to an instance $(\alpha,k)$ of $\univmatch (\alpha,k)$ with $k = 5n + m + 2$, the alphabet $\Sigma \coloneqq \set{\mathtt{0,1,\hashtag,\dollar}}$ and the variable set~$\mathcal{X} \coloneqq \set{\mathtt{z_1,z_2,\ldots z_n, u_1,u_2, \ldots u_n}}$. More precisely, we want to show that there exists a substitution~$h$ to replace all the variables in $\alpha$ with constant words, such that $\iota(h(\alpha)) = 5n + m + 2$, if and only if the boolean formula $\varphi$ is satisfiable. Our construction can be performed in polynomial time and is of polynomial size with respect to $N$. To present this construction, we will go through its building blocks, the so-called gadgets.


Before we start with these gadgets, let us introduce a renaming function for the variables $\rho : X \cup \bar{X} \rightarrow \mathcal{X}$ with $\rho(x_i) = \mathtt{z_i}$ and $\rho(\bar{x}_i) = \mathtt{u_i}$. Also, a substitution $h$ which maps $\alpha$ to a string of universality index $5n+m+2$ is called valid in the following. 

\textbf{The binarisation gadgets.} We use the following gadgets to make the image of variables~$\mathtt{z_i}$ and $\mathtt{u_i}$ under a valid substitution be strings over $\{0,1\}$. Recall that we have the alphabet $\Sigma \coloneqq \set{\mathtt{0,1,\hashtag,\dollar}}$ and the set of variables $\mathcal{X} \coloneqq \set{\mathtt{z_1,z_2,\ldots, z_n, u_1,u_2, \ldots, u_n}}$. 

At first, we construct the gadget $\mathtt{\pi_{\hashtag}}= \mathtt{(z_1z_2\cdots z_n u_1u_2 \cdots u_n 0 1 \dollar)^{N^6} \hashtag}$, as shown in Figure \ref{fig:pihash}. We observe that for all possible substitutions~$h$, we have two cases for the universality of the image of this gadget. On the one hand, assume that any of the variables is substituted under $h$ by a string that contains a $\hashtag$. Then, the universality index of the image of this gadget will be $\iota(h(\mathtt{\pi_{\hashtag}})) = k'$ with $k' \geq N^6 > k$, which is too big for a valid substitution. On the other hand, when all the variables are substituted under $h$ by strings that do not contain $\hashtag$, this gadget is mapped to a string which is exactly one arch because there is only one $\hashtag$ at its very end. Thus, under a valid substitution $h$, the images of the variables $\mathtt{z_i}$ and $\mathtt{u_i}$ do not contain $\hashtag$. Note also that, in the arch factorisation of such a string ($h(\mathtt{\pi_{\hashtag}})$, where $h$ is a valid substitution) we have one arch and no rest. 

The gadget~$\mathtt{\pi_{\dollar}} = \mathtt{(z_1z_2\cdots z_n u_1u_2 \cdots u_n 0 1 \hashtag)^{N^6}}\dollar $  is constructed analogously and can be seen in Figure~\ref{fig:pidollar}. This enforces that under a valid substitution $h$, the images of the variables $\mathtt{z_i}$ and $\mathtt{u_i}$ do not contain $\dollar$.

In conclusion, the gadgets $\mathtt{\pi_{\hashtag}}$ and $\mathtt{\pi_{\dollar}}$ enforce that under a valid substitution $h$, the image of the variables $\mathtt{z_i}$ and $\mathtt{u_i}$ contains only $0$ and $1$, i.e., they are binary strings.

\begin{figure}[!ht]
	\center
	\begin{tikzpicture}[]
        \node[circle] (w0) at (0,0) {$\mathtt{\pi_{\hashtag} = }$};
        \node[minimum height=8mm,align=center, right= 0mm of w0] (w1) {$\mathtt{(}$};
        \node[minimum height=8mm,align=center, right= 0mm of w1] (w2) {$\mathtt{z_1z_2\cdots z_n}$};
        \node[minimum height=8mm,align=center, right= 0mm of w2] (w3) {$\mathtt{u_1u_2\cdots u_n}$};
        \node[minimum height=8mm,align=center, right= 0mm of w3] (w4) {$\mathtt{01\dollar}$};
        \node[minimum height=8mm,align=center, right= 0mm of w4] (w5) {$\mathtt{)^{N^6}}$};
        \node[minimum height=8mm,align=center, right= 0mm of w5] (w6) {$\mathtt{\hashtag}$};
              
        \draw [decorate, thick, decoration = {calligraphic brace, mirror, amplitude=10pt}] (w2.south west) --  (w3.south east) node  [midway, yshift=-0.5cm]  {No $\mathtt{\hashtag}$ allowed};
        
    \end{tikzpicture}
    \caption{If any of the variables is substituted by a string that contains a $\hashtag$, then this gadget would add at least $N^6$ arches, which is already greater than the target universality $k = 5n + m + 2$.}
    \label{fig:pihash}
\end{figure}

\begin{figure}[!ht]
	\center
	\begin{tikzpicture}[]
        \node[circle] (w0) at (0,0) {$\mathtt{\pi_{\dollar} = }$};
        \node[minimum height=8mm,align=center, right= 0mm of w0] (w1) {$\mathtt{(}$};
        \node[minimum height=8mm,align=center, right= 0mm of w1] (w2) {$\mathtt{z_1z_2\cdots z_n}$};
        \node[minimum height=8mm,align=center, right= 0mm of w2] (w3) {$\mathtt{u_1u_2\cdots u_n}$};
        \node[minimum height=8mm,align=center, right= 0mm of w3] (w4) {$\mathtt{01\hashtag}$};
        \node[minimum height=8mm,align=center, right= 0mm of w4] (w5) {$\mathtt{)^{N^6}}$};
        \node[minimum height=8mm,align=center, right= 0mm of w5] (w6) {$\mathtt{\dollar}$};
              
        \draw [decorate, thick, decoration = {calligraphic brace, mirror, amplitude=10pt}] (w2.south west) --  (w3.south east) node  [midway, yshift=-0.5cm]  {No $\mathtt{\dollar}$ allowed};
        
    \end{tikzpicture}
    \caption{If any of the variables is substituted by a string that contains a $\dollar$, then this gadget would add at least $N^6$ arches, which is already greater than the target universality $k = 5n + m + 2$.}
    \label{fig:pidollar}
\end{figure}

\textbf{The Boolean gadgets.} We use the following gadgets to force the image of each $\mathtt{z_i}$ and $\mathtt{u_i}$ to be either in $\mathtt{0}^*$ or $\mathtt{1}^*$. Intuitively, mapping a variable $\mathtt{z_i}$ (respectively, $\mathtt{u_i}$) to a string of the form $0^+$ corresponds to mapping $x_i$ (respectively, $\bar{x}_i$) to the Boolean value false (respectively, true). Similarly, mapping one of these string-variables to a string from $1^+$ means mapping the corresponding boolean variable to true. For a beginning, these gadgets just have to enforce that the image of any string-variable does not contain both $\mathtt{0}$ and $\mathtt{1}$. We construct the gadget $\mathtt{\pi_i^z}$ (respectively $\mathtt{\pi_i^u}$) for every string-variable $\mathtt{z_i}$ (respectively, $\mathtt{u_i}$), according to Figure~\ref{fig:pizoneandzero}. 
More precisely, for all $i\in[n]$, we define two gadgets $\mathtt{\pi_i^z}= (\mathtt{z_i \dollar \hashtag})^{N^6} \mathtt{1001\dollar\hashtag}$ and $\mathtt{\pi_i^u}= (\mathtt{u_i \dollar \hashtag})^{N^6} \mathtt{1001\dollar\hashtag}$.

We now analyse the possible images of $\mathtt{\pi_i^z}= (\mathtt{z_i \dollar \hashtag})^{N^6} \mathtt{1001\dollar\hashtag}$ under various substitutions $h$. There are three ways in which ${\mathtt{z_i}}$ can be mapped to a string by $h$. Firstly, if the image of ${\mathtt{z_i}}$ contains both $\mathtt{0}$ and $\mathtt{1}$, then for the universality index of the image of $\mathtt{\pi_i^z}$ under the respective substitution is $\iota(h(\mathtt{\pi_i^z})) \geq N^6 > k$; such a substitution cannot be valid. Secondly, if the image of ${\mathtt{z_i}}$ is a string from $\mathtt{0}^*$, then the universality of this gadget is exactly $\iota(h(\mathtt{\pi_i^z})) = 2$ as shown in Figure \ref{fig:pizone}. As a third option, if the image of ${\mathtt{z_i}}$ is a string from $\mathtt{1}^*$, then the universality of this gadget is exactly $\iota(h(\mathtt{\pi_i^z})) = 2$ as shown in Figure \ref{fig:pizzero}. As for the binarisation gadgets, in the arch factorisation of a string $h(\mathtt{\pi_i^z})$, where $h$ is a valid substitution, we have exactly two arches (and no rest). 
A similar analysis can be performed for the gadgets $\mathtt{\pi_i^u}= (\mathtt{u_i \dollar \hashtag})^{N^6} \mathtt{1001\dollar\hashtag}$.
In conclusion, the gadgets $\mathtt{\pi_i^z}$  and $\mathtt{\pi_i^u}$ enforce that under a valid substitution $h$, the image of the variables $\mathtt{z_i}$ and $\mathtt{u_i}$ contains either only $0$s or only $1$s (or is empty). 

\begin{figure}[!ht]
	\center
	\begin{tikzpicture}[]
        \node[circle] (w0) at (0,0) {$\mathtt{\pi_i^z = }$};
        \node[minimum height=8mm,align=center, right= 0mm of w0] (w1) {$\mathtt{(}$};
        \node[minimum height=8mm,align=center, right= 0mm of w1] (w2) {$\mathtt{z_i}$};
        \node[minimum height=8mm,align=center, right= 0mm of w2] (w3) {$\mathtt{\dollar\hashtag }$};
        \node[minimum height=8mm,align=center, right= 0mm of w3] (w4) {$\mathtt{)^{N^6}}$};
        \node[minimum height=8mm,align=center, right= 0mm of w4] (w5) {$\mathtt{1}$};
        \node[minimum height=8mm,align=center, right= 0mm of w5] (w6) {$\mathtt{0}$};
        \node[minimum height=8mm,align=center, right= 0mm of w6] (w7) {$\mathtt{01\dollar\hashtag }$};
              
        \draw [decorate, thick, decoration = {calligraphic brace, mirror, amplitude=10pt}] (w2.south west) --  (w3.south east) node  [midway, yshift=-0.5cm]  {No $\mathtt{0}$ and $\mathtt{1}$ together allowed};
   
    \end{tikzpicture}
    \caption{If any of the variables in the gadget $\mathtt{\pi_i^z}$ (analogous for $\mathtt{\pi_i^u}$) is substituted by a string that contains both a $0$ and a $1$, then this gadget would add $N^6$ arches, which is already greater than the target universality $k = 5n + m + 2$.}    
    \label{fig:pizoneandzero}
\end{figure}


\begin{figure}[!ht]
	\center
	\begin{tikzpicture}[]
        \node[circle] (w0) at (0,0) {$\mathtt{\pi_i^z = }$};
        \node[minimum height=8mm,align=center, right= 0mm of w0] (w1) {$\mathtt{(}$};
        \node[minimum height=8mm,align=center, right= 0mm of w1] (w2) {$\mathtt{z_i}$};
        \node[minimum height=8mm,align=center, right= 0mm of w2] (w3) {$\mathtt{\dollar\hashtag }$};
        \node[minimum height=8mm,align=center, right= 0mm of w3] (w4) {$\mathtt{)^{N^6}}$};
        \node[minimum height=8mm,align=center, right= 0mm of w4] (w5) {$\mathtt{1}$};
        \node[minimum height=8mm,align=center, right= 0mm of w5] (w6) {$\mathtt{0}$};
        \node[minimum height=8mm,align=center, right= 0mm of w6] (w7) {$\mathtt{01\dollar\hashtag }$};
              
        \draw [decorate, thick, decoration = {calligraphic brace, mirror, amplitude=10pt}] (w2.south west) --  (w5.south east) node  [midway, yshift=-0.5cm]  {arch if $\mathtt{z_i} \in \mathtt{0}^+$};
        \draw [decorate, thick, decoration = {calligraphic brace, mirror, amplitude=10pt}] (w6.south west) --  (w7.south east) node  [midway, yshift=-0.5cm] {};
        
    \end{tikzpicture}
    \caption{If any of the variables $\mathtt{\pi_i^z}$ (analogous for $\mathtt{\pi_i^u}$) consits of only $0$'s, this gadget would add $2$ arches per variable.}
    \label{fig:pizone}
\end{figure}


\begin{figure}[!ht]
	\center
	\begin{tikzpicture}[]
        \node[circle] (w0) at (0,0) {$\mathtt{\pi_i^z = }$};
        \node[minimum height=8mm,align=center, right= 0mm of w0] (w1) {$\mathtt{(}$};
        \node[minimum height=8mm,align=center, right= 0mm of w1] (w2) {$\mathtt{z_i}$};
        \node[minimum height=8mm,align=center, right= 0mm of w2] (w3) {$\mathtt{\dollar\hashtag }$};
        \node[minimum height=8mm,align=center, right= 0mm of w3] (w4) {$\mathtt{)^{N^6}}$};
        \node[minimum height=8mm,align=center, right= 0mm of w4] (w5) {$\mathtt{1}$};
        \node[minimum height=8mm,align=center, right= 0mm of w5] (w6) {$\mathtt{0}$};
        \node[minimum height=8mm,align=center, right= 0mm of w6] (w7) {$\mathtt{01\dollar\hashtag }$};
              
        \draw [decorate, thick, decoration = {calligraphic brace, mirror, amplitude=10pt}] (w2.south west) --  (w6.south east) node  [midway, yshift=-0.5cm]  {arch if $\mathtt{z_i} \in \mathtt{1}^*$};
        \draw [decorate, thick, decoration = {brace, mirror, amplitude=8pt}] (w7.south west) --  (w7.south east) node  [midway, yshift=-0.5cm] {};
        
    \end{tikzpicture}
    \caption{If any of the variables $\mathtt{\pi_i^z}$ (analogous for $\mathtt{\pi_i^u}$) consits of only $1$'s, this gadget would add $2$ arches per variable.} 
    \label{fig:pizzero}
\end{figure}


\textbf{The complementation gadgets.} The role of these gadgets is to enforce the property that  $\mathtt{z_i}$ and $\mathtt{u_i}$ are not both in $\mathtt{0}^+$ or not both in $\mathtt{1}^+$, for all $i\in[n]$. We construct the gadget $\mathtt{\xi_i}=\mathtt{\dollar z_i u_i \hashtag}$, for every $i\in[n]$, according to Figure \ref{fig:pinegations}. Let us now analyse the image of these gadgets under a valid substitution ($\mathtt{\pi_{\hashtag}}$ and $\mathtt{\pi_{\dollar}}$ are mapped to exactly one arch each, and $\mathtt{\pi_i^z}$  and $\mathtt{\pi_i^u}$ are mapped to exactly two arches each). 
In this case, we observe that $\mathtt{\xi_i}$ is mapped to exactly one complete arch ending on the rightmost symbol $\hashtag$ if and only if the image of one of the variables $\mathtt{z_i}$ and $\mathtt{u_i}$ has at least one $0$ and the image of the other one has at least one $1$. Further, let us consider the concatenation of two consecutive such gadgets $\mathtt{\xi_i\xi_{i+1}}$ and assume that both $\mathtt{z_i}$ and $\mathtt{u_i}$ are mapped to strings over the same letter or at least one of them is mapped to the empty word. In that case, the first arch must close to the right of the $\dollar$ letter in ${\xi_{i+1}}$, hence ${\xi_i\xi_{i+1}}$ could not contain two arches. Thus, the concatenation of the gadgets ${\xi_1\cdots \xi_{n}}$ is mapped to a string which has exactly $n$ arches if and only if each gadget ${\xi_i}$ is mapped to exactly one arch, which holds if and only if the image of one of the variables $\mathtt{z_i}$ and $\mathtt{u_i}$ has at least one $0$ and the image of the other one has at least one $1$. When assembling together all the gadgets, we will ensure that, in a valid substitution, this property holds: $\mathtt{z_i}$ and $\mathtt{u_i}$ are mapped to repetitions of different letters.


\begin{figure}[!ht]
	\center
	\begin{tikzpicture}[]
        \node[circle] (w0) at (0,0) {$\mathtt{\xi_i = }$};
        \node[minimum height=8mm,align=center, right= 0mm of w0] (w1) {$\mathtt{\dollar}$};
        \node[minimum height=8mm,align=center, right= 0mm of w1] (w2) {$\mathtt{z_i}$};
        \node[minimum height=8mm,align=center, right= 0mm of w2] (w3) {$\mathtt{u_i}$};
        \node[minimum height=8mm,align=center, right= 0mm of w3] (w4) {$\mathtt{\hashtag}$};
              
        
    \end{tikzpicture}
    \caption{In order for $\mathtt{\xi_i}$ to contribute an arch, one of $z_i$ and $u_i$ has to be replaced by only $1$'s while the other must consist of only $0$'s.} 
    \label{fig:pinegations}
\end{figure}

\textbf{The clause gadgets.} Let $c_j = (y_j^1 \lor y_j^2 \lor y_j^3)$ be a clause, with $y_j^1,y_j^2,y_j^3 \in X \cup \bar{X}$. We construct the gadget $\mathtt{\delta_j}$ for every clause $c_j$ as $\mathtt{\dollar 0 \rho(y_j^1) \rho(y_j^2) \rho(y_j^3) \hashtag}$, as shown in Figure \ref{fig:piclause}. Now, by all of the properties discussed for the previous gadgets, we can analyse the possible number of arches contained in the image of this gadget under a valid substitution. Firstly, note that if at least one of the variables $\mathtt{\rho(y_j^1),\rho(y_j^2),\rho(y_j^3)}$ is mapped to a string containing at least one $1$, then this gadget will contain exactly one arch ending on its rightmost symbol $\mathtt{\hashtag}$. Now consider the concatenation of two consecutive such gadgets ${\delta_j\delta_{j+1}}$, and assume that all the variables in ${\delta_j}$ are substituted by only $0$s. In this case, the first arch must end to the right of the $\dollar$ symbol in ${\delta_{j+1}}$, hence the string to which $\mathtt{\delta_j\delta_{j+1}}$ is mapped could not contain two arches. The same argument holds if we look at the concatenation of the last complementation gadget and the first clause gadget, e.g. $\xi_n\delta_1$. 

 Thus, the concatenation of the gadgets ${\delta_1\cdots \delta_{m}}$ is mapped to a string which has exactly $m$ arches if and only if each gadget ${\delta_i}$ is mapped to exactly one arch. This holds if and only if at least one of the string-variables occurring in $\delta_i$ is mapped to a string of $1$s. When assembling together all the gadgets, we will ensure that at least one of the variables occurring in each gadget ${\delta_i}$, for all $i\in [m]$, is mapped to a string of $1$s in a valid substitution.


\begin{figure}[!ht]
	\center
	\begin{tikzpicture}[]
        \node[circle] (w0) at (0,0) {$\mathtt{\delta_j = }$};
        \node[minimum height=8mm,align=center, right= 0mm of w0] (w1) {$\mathtt{\dollar}$};
        \node[minimum height=8mm,align=center, right= 0mm of w1] (w2) {$\mathtt{0}$};
        \node[minimum height=8mm,align=center, right= 0mm of w2] (w3) {$\mathtt{\rho(y_j^1)}$};
        \node[minimum height=8mm,align=center, right= 0mm of w3] (w4) {$\mathtt{\rho(y_j^2)}$};
        \node[minimum height=8mm,align=center, right= 0mm of w4] (w5) {$\mathtt{\rho(y_j^3)}$};
        \node[minimum height=8mm,align=center, right= 0mm of w5] (w6) {$\mathtt{\hashtag}$};
              
        
    \end{tikzpicture}
    \caption{In order for $\mathtt{\delta_j}$ to contribute an arch, at least one of $\mathtt{\rho(y_j^1)}$, $\mathtt{\rho(y_j^2)}$ and $\mathtt{\rho(y_j^3)}$ has to be replaced by only $1$'s.}      
    \label{fig:piclause}    
\end{figure}

\textbf{Final Assemblage.} We finish the construction of the pattern $\alpha$ by concatenating all the gadgets. That is, $\alpha=\mathtt{\pi_{\hashtag}}\mathtt{\pi_{\dollar}}\mathtt{\pi_1^z\pi_1^u\pi_2^z\pi_2^u \cdots \pi_n^z\pi_n^u}\mathtt{\xi_1\xi_2 \cdots \xi_n}\mathtt{\delta_1\delta_2 \cdots \delta_m},$ as shown in Figure \ref{fig:pifinal}. 

\begin{figure}[!ht]
	\center
	\begin{tikzpicture}[]
        \node[circle] (w0) at (0,0) {$\alpha = $};
        \node[minimum height=8mm,align=center, right= 0mm of w0] (w1) {$\mathtt{\pi_{\hashtag}}$};
        \node[minimum height=8mm,align=center, right= 0mm of w1] (w2) {$\mathtt{\pi_{\dollar}}$};
        \node[minimum height=8mm,align=center, right= 0mm of w2] (w3) {$\mathtt{\pi_1^z\pi_1^u\pi_2^z\pi_2^u \cdots \pi_n^z\pi_n^u}$};
        \node[minimum height=8mm,align=center, right= 0mm of w3] (w4) {$\mathtt{\xi_1\xi_2 \cdots \xi_n}$};
        \node[minimum height=8mm,align=center, right= 0mm of w4] (w5) {$\mathtt{\delta_1\delta_2 \cdots \delta_m}$};              
        \draw [decorate, thick, decoration = {brace, mirror, amplitude=8pt}] (w1.south west) --  (w1.south east) node  [midway, yshift=-0.5cm]  {$\mathtt{1}$};
        \draw [decorate, thick, decoration = {brace, mirror, amplitude=8pt}] (w2.south west) --  (w2.south east) node  [midway, yshift=-0.5cm]  {$\mathtt{1}$};show equivalence directions
        \draw [decorate, thick, decoration = {calligraphic brace, mirror, amplitude=10pt}] (w3.south west) --  (w3.south east) node  [midway, yshift=-0.5cm]  {$\mathtt{4n}$};
        \draw [decorate, thick, decoration = {calligraphic brace, mirror, amplitude=10pt}] (w4.south west) --  (w4.south east) node  [midway, yshift=-0.5cm]  {$\mathtt{n}$};                
        \draw [decorate, thick, decoration = {calligraphic brace, mirror, amplitude=10pt}] (w5.south west) --  (w5.south east) node  [midway, yshift=-0.5cm]  {$\mathtt{m}$};    
                
    \end{tikzpicture}
    \caption{The concatenation of all gadgets and their respective amount of arches we expect, if we can find a substitution $h$ with $\iota(h(\alpha)) = 5n + m + 2$.}
    \label{fig:pifinal}
\end{figure}

\textbf{The correctness of the reduction.}
We show that there exists a substitution $h$ of the string variables of $\alpha$ with $\iota(h(\alpha)) = 5n + m + 2$ (i.e., a valid substitution) if and only if we can find an assignment for all Boolean-variables occurring in $\varphi$ that satisfy all clauses $c_j \in \varphi$. 

Let us first show that if there is a satisfying assignment for Boolean-variables of $\varphi$ which makes the formula true, then there exists a substitution $h$ of the string-variables of $\alpha$ such that $\iota(h(\alpha)) = 5n + m + 2$. In this case, we can give a canonical substitution $h$ with $h(\rho(x_i)) = 1$ and $h(\rho(\bar{x}_i)) = 0$ if $x_i$ is assigned true, and $h(\rho(x_i)) = 0$ and $h(\rho(\bar{x}_i)) = 1$  if $x_i$ is assigned false. We can easily verify, by the definition of the gadgets, that under this substitution we have $\iota(h(\alpha)) = 5n + m + 2$. Indeed, in the images of each gadget $\mathtt{\pi_{\hashtag}}, \mathtt{\pi_{\dollar}}, \mathtt{\pi^z_i},\xi_i $ and $\delta_i$ we have exactly one arch, ending on the last symbol of the respective strings, while in the image of each gadget $\mathtt{\pi^z_i}$ under this substitution there will be exactly two arches, again ending on their last positions.

Conversely, we want to show that if we have a substitution $h$ of the string-variables such that $\iota(h(\alpha)) = 5n + m + 2$, then there must be a satisfying assignment of the Boolean-variables for $\varphi$. The general idea is the following. We assume that we have a substitution of the string variables and compute the arch factorisation greedily and look at the properties enforced by the individual gadgets, as discussed above. Assume first, towards a contradiction, that the image of some variable~$\mathtt{z_i}$ contains both $0$ and $1$ or that it contains $\hashtag$ or $\dollar$.
Then, as explained, the number of arches of the image of
$\mathtt{\pi_{\hashtag}\pi_{\dollar}\pi_1^z\pi_1^u\pi_2^z\pi_2^u \cdots \pi_n^z\pi_n^u}$
will blow up to a value greater than $5n + m + 2$, a contradiction. The same reasoning holds for the variables $\mathtt{u_i}$. Therefore, each variable $\mathtt{z_i}$ is mapped to a string from $\mathtt{0}^*\cup \mathtt{1}^*$, and the same holds for the variables $\mathtt{u_i}$.
It follows that $h(\mathtt{\pi_{\hashtag}\pi_{\dollar}\pi_1^z\pi_1^u\pi_2^z\pi_2^u \cdots \pi_n^z\pi_n^u})$ contributes exactly $4n + 2$ arches to the arch factorisation of $h(\alpha)$, and the last arch of this factorisation (when identified greedily, from left to right) ends on the last letter of $\mathtt{\pi_n^u}$ (which is a $\hashtag$ symbol). By this last property, we are guaranteed that we can look at the suffix $h(\mathtt{\xi_1\xi_2 \cdots \xi_n\delta_1\delta_2 \cdots \delta m})$ of our pattern's image under $h$ separately, as no arch from the prefix $h(\mathtt{\pi_{\hashtag}\pi_{\dollar}\pi_1^z\pi_1^u\pi_2^z\pi_2^u \cdots \pi_n^z\pi_n^u})$ extends in it. More precisely, this allows us to consider the subproblem of analysing $h$ under the assumption that $\iota(h(\mathtt{\xi_1\xi_2 \cdots \xi_n\delta_1\delta_2 \cdots \delta m})) = m + n$, and, moreover, each string variable is mapped to strings from $\mathtt{0}^*\cup \mathtt{1}^*$. In this subproblem, we have $n+m$ $\dollar$ symbols in the pattern and we can not introduce new $\dollar$ symbols in the image of the string-variables. Therefore, every $\dollar$ symbol needs to be in exactly one arch. Now, as discussed when introducing the complementation and clause gadgets, we have to have the following properties, as otherwise we would have at least two $\dollar$ symbols in the same arch and would only get to $k'<m+n$ arches overall. Firstly, one of each $h(\rho(x_i))$ and $h(\rho(\bar{x}_i))$ has to consist only of $\mathtt{0}$s while the other consists only of $\mathtt{1}$s, and both of them should have length at least $1$. Secondly, at least one of each $\mathtt{\rho(y_j^1)}$, $\mathtt{\rho(y_j^2)}$ and $\mathtt{\rho(y_j^3)}$ has to be substituted by a string from $\mathtt{1}^+$. 

Given these properties, we can construct a satisfying assignment of the Boolean-variables from $\varphi$ by setting a variable to be true if and only if their corresponding string-variable is mapped to a string from $\mathtt{1}^*$. As $h(\rho(x_i))$ and $h(\rho(\bar{x}_i))$ are mapped to strings over distinct alphabets, we get that $x_i$ and $\bar{x}_i$ will be assigned distinct truth values. Moreover, at least one of each $\mathtt{\rho(y_j^1)}$, $\mathtt{\rho(y_j^2)}$ and $\mathtt{\rho(y_j^3)}$ has to be substituted by a string from $\mathtt{1}^+$, so at least one variable per clause is assigned to true. Therefore, this assignment makes $\varphi$ true.

This concludes our proof, and shows that $\univmatch (\alpha,k)$ is \NP-hard.  \qed\end{proof}

In the following we show that $\univmatch (\alpha,k)$ is in \NP.
One natural approach is to guess the images of the variables occurring in the input pattern $\alpha$ under a substitution~$h$
and check whether or not $\universalityIdx{}{h(\alpha)}$ is indeed $k$.
However, it is difficult to bound the size of the images of the variables of $\alpha$ under $h$ in terms of the size of $\alpha$ and $\log k$ (the size of our input), since the strings we look for may be exponentially long. For example, consider the pattern~$\alpha=X_1$: the length of the shortest $k$-universal string is $k\sigma$~\cite{BarkerFHMN20}, which is already exponential in $\log k$.
Therefore, we consider guessing only the subsequence universality signatures for the image of each variable under the substitution. We show that it is sufficient to guess $|\var(\alpha)|$ subsequence universality signatures, one for each variable, instead of the actual images of the variables under a substitution~$h$ using the following proposition by Schnoebelen and Veron~\cite{SchnoebelenV23}.

\begin{proposition}[\!\!\cite{SchnoebelenV23}]\label{prop:signatureIsSufficient}
    For $u,v\in\Sigma^{\ast}$,
    we can compute $\subseqUnivSign{uv}$,
    given the subsequence universality
    signatures~$\subseqUnivSign{u}
    =(\gamma_u,\myKarr{}_u,\myRarr{}_u)$
    and $\subseqUnivSign{v}
    =(\gamma_v,\myKarr{}_v,\myRarr{}_v)$
    of each string, 
    in time polynomial in $|\alphabetOf{uv}|$ and $\log t$,
    where $t$ is the maximum element
    of $\myKarr{}_u$ and $\myKarr{}_v$.
\end{proposition}

Once we have guessed
the subsequence universality signatures
of all variables in $\var(\alpha)$
under substitution~$h$,
we can compute
$\universalityIdx{}{h(\alpha)}$
in the following way.
We first compute the subsequence universality signature
of the maximal prefix of $\alpha$
that does not contain any variables.
We then incrementally compute
the subsequence universality signature
of prefixes of the image of $\alpha$.
Let $\alpha=\alpha_1\alpha_2$,
where we already have
$\subseqUnivSign{h(\alpha_1})$
from induction.
If $\alpha_2[1]$ is a variable,
we compute $\subseqUnivSign{h(\alpha_1\alpha_2[1])}$
from $\subseqUnivSign{h(\alpha_1)}$
and the guessed subsequence universality signature
for variable~$\alpha_2[1]$,
using Proposition~\ref{prop:signatureIsSufficient}.
Otherwise, we take the maximal prefix~$w$ of $\alpha_2$
that does not consist of any variables.
We first compute $\subseqUnivSign{w}$
and then compute $\subseqUnivSign{h(\alpha_1w)}$
using Proposition~\ref{prop:signatureIsSufficient}.
Once we have $\subseqUnivSign{h(\alpha)}
=(\gamma,\myKarr{},\myRarr{})$,
we compute $\universalityIdx{}{h(\alpha)}=\myKarr{}[\sigma]+1$.
Note that the whole process can be done
in a polynomial number of steps 
in $|\alpha|$, $\log k$, and $\sigma$
due to Proposition~\ref{prop:signatureIsSufficient}, provided that the signatures are of polynomial size.

Thus, we now measure the encoding size of a subsequence universality signature and, as such, the overall size of the certificate for $\univmatch$ that we guess.
We can use $\sigma!$ bits to encode a permutation~$\gamma$ of 
a subset of $\Sigma$.
An integer between 1 and $\sigma-1$ requires $\log\sigma$ bits.
Naively, $\myRarr{}$ requires $(2^\sigma)^\sigma$ bits
because there can be $2^\sigma$ choices for each item.
Finally, in the framework of our problem,
note that $\myKarr{}[1]-\myKarr{}[|\gamma|]\le 1$
by Schnoebelen and Veron~\cite{SchnoebelenV23},
and that the values of $\myKarr{}[i]$ are non-increasing in $i$.
Therefore, we can encode $\myKarr{}$
as a tuple~$(l,k')$
where $k'=\max\{\myKarr{}[i]\mid 1\le i\le |\gamma|\}\le k$
and $l=|\{i\in[|\gamma|]\mid \myKarr{}[i]=k'\}|$.
This encoding scheme requires at most $\log\sigma+\log k$ bits.
Summing up, the overall space required to encode a certificate that consists of $|\var(\alpha)|$ subsequence universality signatures
takes at most $(1+\sigma!+(2^\sigma)^\sigma+\log\sigma+\log k)|\var(\alpha)|$ bits.
This is polynomial in the size of the input
and the number of variables, because we assume a constant-sized alphabet, i.e. $\sigma\in O(1)$.

It remains to design a deterministic polynomial algorithm
that tests the validity
of the guessed subsequence universality signature.
Assume that we have guessed the 3-tuple~$(\gamma,\myKarr{},\myRarr{})$.
We claim that there are only constantly many
strings we need to check
to decide whether or not
$(\gamma,\myKarr{},\myRarr{})$
is a valid subsequence universality signature - 
allowing us a brute-force approach.
Lemma \ref{lem:marginalPumpDownAndUp} allows us
to ``pump down''
strings with universality index
greater than $(2^\sigma)^\sigma$,
which is a constant. \looseness=-1

\begin{restatable}[]{lemma}{refmarginalPumpDownAndUp}\label{lem:marginalPumpDownAndUp}
    The tuple~$(\gamma,\myKarr{}_1,\myRarr{})$
    is a valid subsequence universality signature
    iff there exists $w\in\Sigma^{\ast}$
    with $\universalityIdx{}{w}\le (2^\sigma)^\sigma$,
    $\subseqUnivSign{w}=(\gamma,\myKarr{}_2,\myRarr{})$, and
    $\myKarr{}_1[t]-\myKarr{}_2[t]=c\in\N_0$ for all $t\in[|\gamma|]$.
\end{restatable}

\begin{proof}
    \textbf{Only if part.}
    Let $\subseqUnivSign{w}=(\gamma,\myKarr{}_2,\myRarr{})$
    with $\universalityIdx{}{w}\ge 1$.
    Then, we have $\subseqUnivSign{\gamma^cw}=(\gamma,\myKarr{}_1,\myRarr{})$
    where $\myKarr{}_1[t]=\myKarr{}_2[t]+c$ for all $t\in[\sigma]$
    and all $c\in\N_0$.
    
    \textbf{If part.}
    Let $u$ be a string
    with $\subseqUnivSign{u}=(\gamma,\myKarr{}_1,\myRarr{})$.
    If $\universalityIdx{}{u}\le (2^\sigma)^\sigma$,
    the statement is already true.
    Otherwise, we have $\universalityIdx{}{u}>(2^\sigma)^\sigma$.
    Consider two integers~$i$ and $j$ $(i<j)$,
    which are multiples of $\sigma$.
    Assume that $\alphabetOf{u[\marginalSeq{i+l}{u}+1:\marginalSeq{i+l+1}{u}]}
    =\alphabetOf{u[\marginalSeq{j+l}{u}+1:\marginalSeq{j+l+1}{u}]}$
    for $l\in[\sigma-1]\cup\{0\}$.
    Note that the endpoints of the arches
    after $\marginalSeq{j+\sigma}{u}$
    depend exactly on the set of characters~$\alphabetOf{
    w[\marginalSeq{j+l}{u}+1:\marginalSeq{j+l+1}{u}]}$
    from $l=0$ to $\sigma-1$,
    and the suffix~$u[\marginalSeq{j+\sigma}{u}]$.
    Therefore, we can remove
    $u[\marginalSeq{i+\sigma}{u}+1:\marginalSeq{j+\sigma}{u}]$
    without altering $\gamma$ or $\myRarr{}$.
    This argument is illustrated in Figure~\ref{fig:marginalDownPump}.

    \begin{figure}[!h]
	\centering
        \includegraphics[scale=1]{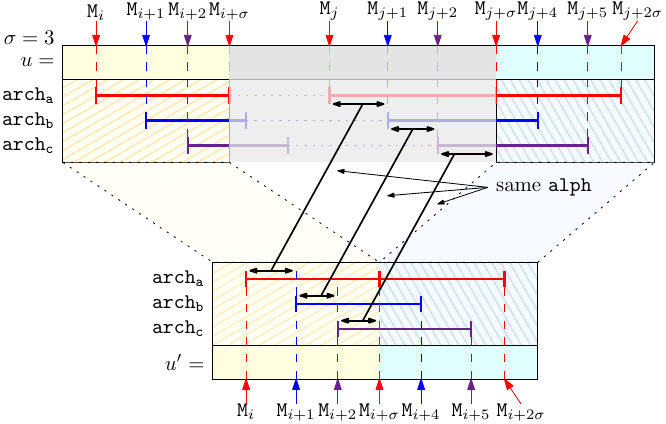}
    \caption{We can safely remove
    the substring~$u[\marginalSeq{i+\sigma}{u}+1:
    \marginalSeq{j+\sigma}{u}]$}
    to obtain $u'$
    with the same $\gamma$ and $\myRarr{}$,
    and all $\myKarr{}$ values are lower
    by $\frac{j-i}{\sigma}$.
    \label{fig:marginalDownPump}
    \end{figure}
    
    Now, let $\myKarr{}_2$ be the array of integers
    obtained by subtracting $\frac{j-i}{\sigma}$
    from all values in $\myKarr{}_1$.
    Then, $\subseqUnivSign{u[1:\marginalSeq{i+\sigma}{u}]
    u[\marginalSeq{j+\sigma}{u}+1:|u|]}=(\gamma,\myKarr{}_2,\myRarr{})$.
    Since there can be na\"ively $2^\sigma$ choices of
    $\alphabetOf{u[\marginalSeq{i+l}{u}+1:\marginalSeq{i+l+1}{u}]}$
    for each $l\in[\sigma-1]\cup\{0\}$,
    any string with $\universalityIdx{}{u}>(2^\sigma)^\sigma$
    is guaranteed to have integers~$i$ and $j$
    that satisfy the above conditions
    by the pigeonhole principle.
    Therefore, we will reach a string~$w$
    with $\universalityIdx{}{w}\le (2^\sigma)^\sigma$
    and $\subseqUnivSign{w}=(\gamma,\myKarr{}_2,\myRarr{})$,
    where $\myKarr{}_1[t]-\myKarr{}_2[t]=c$
    for some non-negative constant~$c$
    and all $t\in[\sigma]$
    if we repeatedly apply the same argument
    to remove arches.
    \qed
\end{proof}

Lemma~\ref{lem:marginalPumpDownAndUp} limits the search space for the candidate string corresponding to a tuple $(\gamma,\myKarr{},\myRarr{})$ 
by mapping valid subsequence universality signatures
to subsequence universality signatures
for strings with universality index at most $(2^\sigma)^\sigma$.
Therefore, we need to investigate those strings where there are up to $\sigma\cdot(1+(2^\sigma)^\sigma)+1$ terms in its marginal sequence.
The following lemma bounds the length of the substring between two consecutive marginal sequence terms in such a string. The conclusion of this line of thought follows then, in Corollary \ref{cor:searchSpaceLengthBound}. \looseness=-1

\begin{restatable}[]{lemma}{refoncePerMargin}\label{lem:oncePerMargin}
    For a given string~$w$,
    let $w=uvx$
    where $v=w[\marginalSeq{i}{w}+1:\marginalSeq{i+1}{w}]\neq \varepsilon$,
    and $u=w[1:\marginalSeq{i}{w}]$,
    and $x=w[\marginalSeq{i+1}{w}+1:|w|]$
    for some integer~$i\ge 1$.
    For a permutation~$v'$ of $\alphabetOf{v}$
    that ends with $v[|v|]$,
    we have $\subseqUnivSign{uvx}
    =\subseqUnivSign{uv'x}$.
\end{restatable}

\begin{proof}
    Since $v$ is between two consecutive marginal sequence terms,
    all arches for any signature letter in $uvx$
    must end at a position no more than $|u|$
    or at a position no less than $|uv|$
    Arches that end at a position no more than $|u|$ in $uvx$
    will end at the same positions in $uv'x$
    because they only depend on $u$.
    Suppose that an arch for signature letter~$\mathtt{a}$
    starts at a position no more than $|u|$
    and ends at a position no less than $|uv|$ in $uvx$.
    Let $i$ be the minimum non-negative integer
    that allows $\alphabetOf{\rest{\mathtt{a}}{u}vx[1:i]}=\Sigma$.
    If $i\ge 1$,
    we have $\alphabetOf{\rest{\mathtt{a}}{u}v}\ne\Sigma$.
    Since $\alphabetOf{\rest{\mathtt{a}}{u}v}=\alphabetOf{\rest{\mathtt{a}}{u}v'}$,
    the minimum integer~$i'$
    that allows $\alphabetOf{\rest{\mathtt{a}}{u}v'x[1:i']}=\Sigma$
    is equal to $i$.
    On the other hand, if $i=0$,
    then we have $\alphabetOf{\rest{\mathtt{a}}{u}v}=\Sigma$
    and $\alphabetOf{\rest{\mathtt{a}}{u}v[1:|v|-1]}\ne\Sigma$.
    Since we have $v'[|v'|]=v[|v|]$
    and $\alphabetOf{v'[1:|v'|-1]}=\alphabetOf{v[1:|v|-1]}$,
    the arch ends exactly at $|uv'|$ in $uv'x$.
    The number of arches that continues afterwards
    and the corresponding letters in the rest
    are thus equal for $uvx$ and $uv'x$.
    Finally, even if $v$ is in the rest
    of the arch factorization for a signature letter~$\mathtt{a}$,
    we still have $\alphabetOf{\rest{\mathtt{a}}{uvx}}
    =\alphabetOf{\rest{\mathtt{a}}{uv'x}}$
    because $\alphabetOf{v}=\alphabetOf{v'}$.
    \qed
\end{proof}

\begin{corollary}\label{cor:searchSpaceLengthBound}
    The tuple~$(\gamma,\myKarr{},\myRarr{})$
    is a valid subsequence universality signature 
    if and only if there exists a string~$w$ of length
    at most $\sigma\cdot(\sigma\cdot(1+(2^\sigma)^\sigma)+1)$ and a constant $c\in \N_0$
    that satisfies $\subseqUnivSign{w}=(\gamma,\myKarr{}-c,\myRarr{})$.
\end{corollary}

We can now show the following result.
\begin{restatable}[]{lemma}{refmatchUnivNP}\label{lem:matchUnivNP}
    $\univmatch (\alpha,k)$ is in \NP.
\end{restatable}

\begin{proof}
Follows from Proposition~\ref{prop:signatureIsSufficient} and Corollary~\ref{cor:searchSpaceLengthBound}.
Firstly, for a guessed sequence of universality signatures~$(\gamma_x,\myKarr{}_x,\myRarr{}_x)$, for $x\in \var(\alpha)$, we check their validity. For that, we enumerate all strings of length up to the constant~$\sigma\cdot(\sigma\cdot(1+(2^\sigma)^\sigma)+1)$ over $\Sigma$ and see if there exist strings~$w_x$ such that $\subseqUnivSign{w_x}=(\gamma_x,\myKarr{}_x-c_x,\myRarr{}_x)$ for some constant $c_x\leq k$. Since $\sigma$ is constant, this takes polynomial time. We then use Proposition \ref{prop:signatureIsSufficient} to check if the guessed signatures lead to an assignment $h$ of the variables such that $\iota(h(\alpha))=k$, as already explained. Since we have a polynomial size bound on the certificate and a deterministic verifier that runs in polynomial time, we obtain that $\univmatch (\alpha,k)$ is in \NP.
\qed
\end{proof}

Based on Lemmas~\ref{lemma:matchUnivnphard} and \ref{lem:matchUnivNP}, the following theorem follows.
\begin{theorem}
    $\univmatch$ is \NP-complete.
\end{theorem}

Further, we describe two classes of patterns, defined by structural restrictions on the input patterns, for which $\univmatch$ can be solved in polynomial time. 

\begin{restatable}[]{proposition}{refoneVarOneOcc}\label{rem:oneVarOneOcc}
    {\upshape a)} $\univmatch (\alpha,k)$ is in \P~when there exists a variable
    that occurs only once in $\alpha$. As such, $\univmatch (\alpha,k)$ is in \P~for the heavily studied class of regular patterns (see, e.g., \cite{DBLP:journals/toct/FernauMMS20} and the references therein), where each variable occurs only once.
    {\upshape b)} $\univmatch (\alpha,k)$ is in \P~when
    $|\var(\alpha)|$ is constant.
\end{restatable}

\begin{proof}
a)    Let $x$ be the variable that occurs only once in $\alpha$.
    Then, we can uniquely rewrite $\alpha=\alpha_1 x \alpha_2$.
    We will successively define three substitutions~$h_1$, $h_2$, $h_3$,
    all of which map variables that are not $x$
    to the empty string,
    i.e., $h_1(x')=h_2(x')=h_3(x')=\emptyword$
    for all $x'\in \mathcal{X}\setminus\{x\}$.
    Now, let $h_1(x)=\emptyword$ as well.
    We claim that $k\ge\universalityIdx{}{h_1(\alpha)}$
    if and only if $\univmatch (\alpha,k)$ is true.
    For any substitution~$h$,
    we have $\universalityIdx{}{h_1(\alpha)}\le
    \universalityIdx{}{h(\alpha)}$
    because $\subseq{h_1(\alpha)}{h(\alpha)}$.
    Therefore, the problem is false
    if $k<\universalityIdx{}{h_1(\alpha)}$.
    Moreover, if $k=\universalityIdx{}{h_1(\alpha)}$,
    the problem is true by definition.
    Now, assume $k>\universalityIdx{}{h_1(\alpha)}$
    and let $h_2(x)$ be a permutation of
    $\Sigma\setminus\rest{}{h_2(\alpha_1)}$.
    Then, $\universalityIdx{}{h_2(\alpha)}
    =\universalityIdx{}{h_2(\alpha_1)}+1+\universalityIdx{}{h_2(\alpha_2)}$,
    because $\rest{}{h_2(\alpha_1x)}=\emptyword$.
    Note that we either have $\universalityIdx{}{h_2(\alpha)}
    =\universalityIdx{}{h_1(\alpha)}$
    or $\universalityIdx{}{h_1(\alpha)}+1$.
    Finally, for an integer~$i=k-\universalityIdx{}{h_2(\alpha)}$,
    let $h_3(x)=h_2(x)\gamma^i$
    where $\gamma$ is a permutation of $\Sigma$.
    We now have $\universalityIdx{}{h_3(\alpha)}
    =\universalityIdx{}{h_3(\alpha_1)h_2(x)}+i
    +\universalityIdx{}{h_3(\alpha_2)}
    =i+\universalityIdx{}{h_2(\alpha)}$.
    Thus, for any $k\ge\universalityIdx{}{h_1(\alpha)}$,
    there exists a substitution~$h$
    such that $k=\universalityIdx{}{h(\alpha)}$.
    We therefore compute $\universalityIdx{}{h_1(\alpha)}$,
    the universality index of the image,
    and then return true if and only if
    $\universalityIdx{}{h_1(\alpha)}\le k$,
    which can be done in polynomial time.
    \qed



b)    The subsequence universality signature~$\subseqUnivSign{h(x_i)}$
    of the image of some variable~$x_i\in\var(\alpha)$
    under substitution~$h$
    consists of three items,
    a permutation~$\gamma_i$ of a subset of $\Sigma$,
    an array~$\myKarr{}_i$ of $\sigma$ integers,
    and an array~$\myRarr{}_i$ of $\sigma$ subsets of $\Sigma$.
    Recall from the size estimation of such a universality signature
    that $\myKarr{}_i$ can be represented with two integers~$l_i$ and $k_i$,
    where $\myKarr{}_i[j]=k_i$ for all $j\in[l_i]$
    and $\myKarr{}_i[j]=k_i-1$ for all $j\in[l_i+1:|\gamma_i|]$.
    Note that there are $\sum_{j=0}^{\sigma}j!$ choices for $\gamma$,
    at most $\sigma$ choices for $l_i$,
    and $(2^\sigma)^\sigma$ choices for $\myRarr{}$.
    Therefore, if we treat $k_i$
    as a variable whose value should be determined,
    we can enumerate for all possible assignments
    of $\gamma_i$, $l_i$, and $\myRarr{}_i$
    for all $i\in [|\var(\alpha)|]$
    in constant time
    under the assumption that $\sigma$ and $|\var(\alpha)|$
    are constant.
    
    Now, for a fixed set of $\gamma_i$s, $l_i$s,
    and $\myRarr{}_i$s,
    we find the minimum value~$k_i'$ of $k_i$
    that validates $(\gamma_i,\myKarr{}_i,\myRarr{}_i)$
    as a subsequence universality signature
    by enumerating all strings up to
    length~$\sigma\cdot(\sigma\cdot(1+(2^\sigma)^\sigma)+1)$
    using Corollary~\ref{cor:searchSpaceLengthBound}.
    If no such $k_i'$ exists, we move on to the next set
    of $\gamma_i$s, $l_i$s, and $\myRarr{}_i$s.
    Since Lemma~\ref{lem:marginalPumpDownAndUp}
    allows us to add an arbitrary number of arches
    while not altering $\gamma_i$ and $\myRarr{}_i$,
    we first assume that the number of additional arches is zero
    and compute how many more arches we need for $h(\alpha)$
    to reach a universality index of $k$.
    We compute this number by counting the number of arches
    through an arch factorization on $\alpha$.
    Specifically, for each rewriting~$\alpha_1x_i\alpha_2$,
    we compute the minimal $j\in[|\gamma|]$
    that allows $\alphabetOf{\rest{}{h(\alpha_1)}}\cup\alphabetOf{\gamma[1:j]}=\Sigma$.
    If no such $j$ exists,
    then we simply compute
    $\alphabetOf{\rest{}{h(\alpha_1x_i)}}
    =\alphabetOf{\rest{}{h(\alpha_1)}}\cup\alphabetOf{\gamma}$
    without incrementing the number of arches.
    Then, if $j\le l_i$, we add $k_i'$ arches to the total arch count.
    If $j > l_i$, we add $k_i'-1$ arches instead.
    Finally, we continue the arch factorization process
    with $\alphabetOf{\rest{}{h(\alpha_1x_i)}}=\myRarr{}_i[j]$.
    This way, we can compute
    the minimum number of arches the image of $\alpha$ can have
    for a fixed set of $\gamma_i$s, $l_i$s, and $\myRarr{}_i$s.
    
    Let $d$ be the total number of additional arches we need
    for the image of $\alpha$
    to reach a universality index of $k$.
    Note that an additional arch for each variable~$x_i$
    will contribute to $|\alpha|_{x_i}$ more arches
    in the image of $\alpha$.
    However, if the subsequence universality signature
    features $\gamma_i$ with $\alphabetOf{\gamma_i}\ne \Sigma$,
    then we cannot add any more arches for each variable.
    Let $I=\{i\in[|\var(\alpha)|]\mid |\gamma_i|=\sigma\}$.
    Now, the problem boils down to finding
    how many additional arches we need
    for the image of each variable~$x_i$
    with $i\in I$
    while making the universality index
    of the image of $\alpha$ exactly $k$.
    Let $d_i$ be the number of additional arches
    for each occurrence of variable~$x_i$
    with $i\in I$.
    We can solve for $d_i$s
    the following system of linear inequalities:
    \begin{flalign*}
    &\sum_{i\in I}|\alpha|_{x_i}d_i\le d,\\
    &\sum_{i\in I}-|\alpha|_{x_i}d_i\le -d,\\
    &-d_i\le 0~\forall i\in I
    \end{flalign*}
    Note that the first two inequalities
    imply $\sum_{i=1}^{|\var(\alpha)|}|\alpha|_{x_i}d_i=d$
    and the last $|\var(\alpha)|$ inequalities
    enforce positive values for each $d_i$.
    If there is an integer solution for the system,
    then we can assign $d_i$ more arches for the image of $x_i$,
    and the universality index of the image of $\alpha$
    will be exactly $k$.
    Because $|\var(\alpha)|$ is a constant,
    this system can be solved in time polynomial in $\log H$,
    where $H$ is the maximum between $k$
    and the greatest coefficient of a variable
    in the above system (in absolute value)~\cite{Lenstra83}.
    \qed
\end{proof}

\section{$\simmatch$}\label{section:patternMatchingSimk}

Further, we discuss the $\simmatch$ problem. In the case of $\simmatch$, we are given a pattern $\alpha$, a word $w$, and a natural number $k\leq n$, and we want to check the existence of a substitution $h$ with $h(\alpha)\sim_k w$. The first result is immediate: $\simmatch$ is \NP-hard, because $\simmatch(\alpha,w,|w|)$ is equivalent to $\match(\alpha,w)$, and $\match$ is \NP-complete.

\begin{lemma}\label{lem:simmatchNPH}
$\simmatch$ is \NP-hard.
\end{lemma}
\begin{proof}
    We note that $\simmatch(\alpha,w,|w|)$ is equivalent to the \NP-complete $\match(\alpha,w)$.
\qed\end{proof}

To understand why this results followed much easier than the corresponding lower bound for $\univmatch$, we note that in $\simmatch$ we only ask for $h(\alpha)\sim_k w$ and allow for $h(\alpha)\sim_{k+1} w$, while in $\univmatch$ $h(\alpha)$ has to be $k$-universal but not $(k+1)$-universal. So, in a sense, $\simmatch$ is not strict, while $\univmatch$ is strict. So, we can naturally consider the following problem. 

\begin{problemdescription}
  \problemtitle{Matching under Strict Simon's Congruence: $\simStrictMatch(\alpha,w,k)$}
  \probleminput{Pattern $\alpha$, $|\alpha|=m$, word $w$, $|w|=n$, and $k\in [n]$.}
  \problemquestion{Is there a substitution $h$ with $h(\alpha) \sim_k w$ and $h(\alpha) \not\sim_{k+1} w$?}
\end{problemdescription}

Adapting the reduction from Lemma \ref{lemma:matchUnivnphard}, we can show that $\simStrictMatch$ is \NP-hard.

\begin{restatable}[]{lemma}{refsimStrictMatchNPH}\label{lem:simStrictMatchNPH}
$\simStrictMatch$ is \NP-hard.
\end{restatable}

\begin{proof}
    We refer to the notations from Lemma \ref{lemma:matchUnivnphard}. We use the same reduction from $\sat$ and note that $\alpha$ can either be mapped to a string $h(\alpha)$ with $\iota(h(\alpha))\leq 5n + m + 2$ (with equality only if the input instance of $\sat$ is satisfiability) or to a string $h(\alpha)$ with $\iota(h(\alpha))\leq (n+m)^6$. Therefore, consider the instance of $\simStrictMatch$ with input the pattern $\alpha$ constructed in the reduction, $w=(\mathtt{10\dollar\hashtag})^{5n+m+3}$, and $k=5n+m+2$. Clearly, there exists a substitution $h$ with $h(\alpha) \sim_k w$ and $h(\alpha) \not\sim_{k+1} w$ if and only if there exists a substitution $h$ with $\iota(h(\alpha))=k$. Such a substitution exists if and only if the given instance of $\sat$ is satisfiable.
    \qed
\end{proof}

We can also show an \NP-upper bound: it is enough to consider as candidates for the images of the variables under the substitution $h$ only strings of length $O((k+1)^\sigma)$; longer strings can be replaced with shorter, $\sim_k$-congruent ones, which have the same impact on the sets $\subseqset{k}{h(\alpha)}$. The following holds.
\begin{restatable}[]{theorem}{refsimStrictMatchNPC}\label{thm:simStrictMatchNPC}
$\simmatch$ and $\simStrictMatch$ are \NP-complete.
\end{restatable}

\begin{proof}
By Lemmas \ref{lem:simmatchNPH} and \ref{lem:simStrictMatchNPH}, it is enough to show that both problems are in \NP.

We make some observations first. Note that $\subseqset{k}{w_1w_2}=\Sigma^{\le k}
\cap\subseqset{k}{w_1}\subseqset{k}{w_2}$. Thus, for a pattern~$\alpha$ and two substitutions~$h_1$ and $h_2$ where $h_1(x)\sim_kh_2(x)$ for all variables~$x\in {\mathcal X}$, we have $w\sim_kh_1(\alpha)$ if and only if $w\sim_kh_2(\alpha)$. Moreover, Kim et al.~\cite{KimHKS22} showed that, for a given string~$w$, the length of the shortest string in the set~$\{u\in\Sigma^*\mid u\sim_kw\}$
is at most $\binom{k+\sigma}{\sigma}\le k^{\sigma}$.

Based on these observations, we can now give \NP-algorithms for both problems. We note that these problems reduce to the normal pattern matching problem when $k\ge |w|$. For $\simmatch$, if $k\geq |w|$, we answer $\simmatch(\alpha,w,k)$ positively if and only if $\alpha$ matches $w$. For $\simStrictMatch$, if $k\geq |w|$, we always answer $\simStrictMatch(\alpha,w,k)$ negatively. Indeed, if there exists $h$ such that $h(\alpha)\sim_k w$, then $h(\alpha)\sim_{|w|} w$. It follows that $h(\alpha)=w$ and $h(\alpha)\sim_{k+1} w$, as well; the answer to $\simStrictMatch(\alpha,w,k)$ should therefore be no. 
 
Hence, from now on, we can assume that $k<|w|$. Let us consider first the problem $\simStrictMatch$. From the observations we have made at the beginning of this proof, and taking into account that we need to consider strings congruent under $\sim_{k+1}$, we can conclude that there exists a substitution $h$ such that $h(\alpha)\sim_k w$ and $h(\alpha)\not\sim_{k+1} w$ if and only if there exists such a substitution $h$ where the length of the image of each variable is $(k+1)^{\sigma}$. 

Since $k$ is at most $|w|$ and $\sigma$ is a constant, we only need to test certificates of polynomial length which encode the substitution. The verifier can then simply substitute the variables in the pattern, which will yield a string of length at most~$(k+1)^\sigma |\alpha|$. Now, we can compute the largest $\ell$ for which $h(\alpha)\sim_{\ell} w$ in $O(|h(\alpha)|+|w|) =O((|w|+1)^\sigma|\alpha|+|w|)$ time~\cite{GawrychowskiKKM21}, which is polynomial in the size of the input, under the assumption that $\sigma$ is constant. If $\ell=k$, then we answer the respective instance positively. 
Therefore, the problem is in \NP.

A similar argument holds for $\simmatch$ (but, in that case, it is enough to look for substitutions where the image of the variables is at most $k^\sigma$, as we only deal with $\sim_k$).
On the other hand, note that the same argument cannot be applied to $\univmatch$, because there is no bound on the size of $k$. This makes the size of the certificate, $k^\sigma$, exponentially large in $\log k$, which is the size of the encoding for a binary representation of $k$.

\qed
\end{proof}

Finally, note that $\simmatch$ and $\simStrictMatch$ are in \P~when the input pattern is regular.
\begin{restatable}[]{proposition}{refsimStrictregular}\label{prop:simStrictregular}
If $\alpha$ is a regular pattern, then both problems
    $\simmatch(\alpha,w,k)$ and $\simStrictMatch(\alpha,w,k)$ are in \P.
\end{restatable}

\begin{proof}
    We consider the problem $\simmatch(\alpha,w,k)$. We assume that $\alpha$ is a regular pattern $\alpha = w_0x_1 w_1\cdots x_\ell w_\ell$, where, for $i\in [\ell]$, $x_i$ is a variable and, for $i\in [\ell]\cup\{0\}$, $w_i$ is a string of constants. The language $L(\alpha)$ of all words which can be obtained by replacing the variables of $\alpha$ by constant strings is regular, and we can construct in polynomial time a non-deterministic finite automaton $N_\alpha$ accepting it (it is the automaton accepting the language described by the regular expression $w_0\Sigma^* w_1\cdots \Sigma^* w_\ell$). Now, using the results of \cite{KimHKS22}, we can construct in polynomial time (when the size of the input alphabet $\sigma$ is constant) a deterministic finite automaton $D_{w,k}$ accepting the words which are $\sim_k$ equivalent to $w$. Now, we simply check if there is a word accepted by both these automata ($N_\alpha$ and $D_{w,k}$), which can be done in polynomial time. We return the answer to this check as the answer to $\simmatch(\alpha,w,k)$.

    Further, we consider the problem $\simStrictMatch(\alpha,w,k)$. Just as before, we construct the NFA $N_\alpha$ and the DFA $D_{w,k}$. Moreover, we construct the DFA $D_{w,k+1}$ and its complement $D'_{w,k+1}$ (which accepts the words which are not $\sim_{k+1}$ equivalent to $w$). Now, we see if there is a word accepted by $N_\alpha$ and $D_{w,k}$ and $D'_{w,k+1}$. Clearly, all steps can be done in polynomial time. We return the answer to this check as the answer to the problem $\simStrictMatch(\alpha,w,k)$.
    \qed
\end{proof}

\section{$\simWE$}\label{section:wordEqSimk}

In this section, we address the $\simWE$ problem, where we are given two patterns $\alpha$ and $\beta$, and a natural number $k\leq n$, and we want to check the existence of a substitution $h$ with $h(\alpha)\sim_k h(\beta)$. The first result is immediate: this problem is \NP-hard because $\simmatch$, which is a particular case of $\simWE$, is \NP-hard. 

To show that the problem is in \NP, we need a more detailed analysis. If $k\leq |\alpha|+|\beta|$, the same proof as for the \NP-membership of $\simmatch$ works: it is enough to look for substitutions of the variables with the image of each variable having length at most $k^\sigma$, and this is polynomial in the size of the input. If $k>|\alpha|+|\beta|$, and $\beta=w$ contains no variable, then this is an input for $\simmatch$ with $k$ greater than the length of the input word $w$, and we have seen previously how this can be decided. Finally, if both $\alpha$ and $\beta$ contain variables, then the problem is trivial, irrespective of $k$: the answer to any input is positive, as we simply have to map all variables to $(1\cdots \sigma)^k$ and obtain two $\sim_k$-congruent words. Therefore, we have the following result. 
\begin{theorem}
$\simWE$ is \NP-complete.
\end{theorem}

To avoid the trivial cases arising in the above analysis for $\simWE$, we can also consider a stricter variant of this problem:
\begin{problemdescription}
  \problemtitle{Word Equations under Strict Simon's Congruence: $\simStrictWE(\alpha,\beta,k)$}
  \probleminput{Patterns $\alpha$, $\beta$, $|\alpha|=m$, $\beta=n$, and $k\in [m+n]$.}
  \problemquestion{Is there a substitution $h$ with $h(\alpha) \sim_k h(\beta)$ and $h(\alpha) \not\sim_{k+1} h(\beta)$?}
\end{problemdescription}

Differently from $\simWE$, we can show that this problem is \NP-hard, even in the case when both sides of the pattern contain variables. 
\begin{restatable}[]{lemma}{refsimStrictWENPH}\label{lem:simStrictWENPH}
$\simStrictWE$ is \NP-hard, even if both patterns contain variables.
\end{restatable}

\begin{proof}
    We refer to the notations from Lemma \ref{lemma:matchUnivnphard}. We use the same reduction from $\sat$ and note that $\alpha$ can either be mapped to a string $h(\alpha)$ with $\iota(h(\alpha))\leq 5n + m + 2$ (with equality only if the input instance of $\sat$ is satisfiability) or to a string $h(\alpha)$ with $\iota(h(\alpha))\leq (n+m)^6$. Therefore, consider the instance of $\simStrictWE$ with input the pattern $\alpha$ constructed in the reduction, the second pattern $\beta=(\mathtt{10\dollar\hashtag})^{5n+m+3} x$, where $x$ is a fresh string-variable, and $k=5n+m+2$. Clearly, there exists a substitution $h$ with $h(\alpha) \sim_k h(\beta)$ and $h(\alpha) \not\sim_{k+1} h(\beta)$ if and only if there exists a substitution $h$ with $\iota(h(\alpha))=k$. Such a substitution exists if and only if the given instance of $\sat$ is satisfiable.
    \qed
\end{proof}

Regarding the membership in \NP: if $k$ is upper bounded by a polynomial function in $|\alpha|+|\beta|$ (or, alternatively, if $k$ is given in unary representation), then the fact that $\simStrictWE$ is in \NP~follows as in the case of $\simStrictMatch$. The case when $k$ is not upper bounded by a polynomial in $|\alpha|+|\beta|$ remains open. We can show the following theorem.\looseness=-1
\begin{theorem}\label{lem:simStrictWENPC}
$\simStrictWE$ is \NP-complete, for $k\leq |\alpha|+|\beta|$. 
\end{theorem}

\section{Conclusions}
In this paper, we have considered the problem of matching patterns with variables under Simon's congruence. More precisely, we have considered three main problems $\univmatch$, $\simmatch$, and $\simWE$ and we have given a rather comprehensive image of their computational complexity. These problems are \NP-complete, in general, but have interesting particular cases which are in \P. Interestingly, our \NP~or \P~algorithms work in (non-deterministic) polynomial time only in the case of constant input alphabet (their complexity being, in fact, exponential in the size $\sigma$ of the input alphabet). It seems very interesting to characterize the parameterised complexity of these problems w.r.t. the parameter $\sigma$. In the light of Proposition \ref{rem:oneVarOneOcc}, another interesting parameter to be considered in such a parameterised complexity analysis would be the number of variables. We conjecture that the problems are $W[1]$-hard with respect to both these parameters. \looseness=-1


\newpage

\bibliographystyle{splncs04}
\bibliography{references}

\begin{thebibliography}{10}
\providecommand{\url}[1]{\texttt{#1}}
\providecommand{\urlprefix}{URL }
\providecommand{\doi}[1]{https://doi.org/#1}

\bibitem{adamson2023words}
Adamson, D.: Ranking and unranking $k$-subsequence universal words. In: Frid,
  A., Merca{\c{s}}, R. (eds.) WORDS. pp. 47--59. Springer Nature Switzerland
  (2023)

\bibitem{Goettingen2023words}
Adamson, D., Kosche, M., Ko{\ss}, T., Manea, F., Siemer, S.: Longest common
  subsequence with gap constraints. In: Frid, A., Merca{\c{s}}, R. (eds.)
  WORDS. pp. 60--76 (2023)

\bibitem{Amadini2020}
Amadini, R.: A survey on string constraint solving. ACM Computing Surveys
  (CSUR)  \textbf{55}(1),  1--38 (2021)

\bibitem{ami:gen}
Amir, A., Nor, I.: Generalized function matching. J. Discrete Algorithms
  \textbf{5},  514--523 (2007)

\bibitem{DBLP:journals/jcss/Angluin80}
Angluin, D.: Finding patterns common to a set of strings. J. Comput. Syst. Sci.
   \textbf{21}(1),  46--62 (1980)

\bibitem{BarkerFHMN20}
Barker, L., Fleischmann, P., Harwardt, K., Manea, F., Nowotka, D.: Scattered
  factor-universality of words. In: Jonoska, N., Savchuk, D. (eds.) {DLT} 2020,
  Proceedings. Lecture Notes in Computer Science, vol. 12086, pp. 14--28.
  Springer (2020)

\bibitem{Bruijn46}
de~Bruijn, N.G.: A combinatorial problem. Koninklijke Nederlandse Akademie v.
  Wetenschappen  \textbf{49},  758--764 (1946)

\bibitem{cam:afo}
C\^ampeanu, C., Salomaa, K., Yu, S.: A formal study of practical regular
  expressions. Int. J. Found. Comput. Sci.  \textbf{14},  1007--1018 (2003)

\bibitem{ChenKMS17}
Chen, H.Z.Q., Kitaev, S., M{\"u}tze, T., Sun, B.Y.: On universal partial words.
  Electronic Notes in Discrete Mathematics  \textbf{61},  231--237 (2017)

\bibitem{crochemore}
Crochemore, M., Hancart, C., Lecroq, T.: Algorithms on strings. Cambridge
  University Press (2007)

\bibitem{CrochemoreMT03}
Crochemore, M., Melichar, B., Tron{\'{\i}}cek, Z.: Directed acyclic subsequence
  graph - overview. J. Discrete Algorithms  \textbf{1}(3-4),  255--280 (2003)

\bibitem{day2021edit}
Day, J., Fleischmann, P., Kosche, M., Ko{\ss}, T., Manea, F., Siemer, S.: The
  edit distance to k-subsequence universality. In: {STACS}. vol.~187, pp.
  25:1--25:19 (2021)

\bibitem{DayFMN17}
Day, J.D., Fleischmann, P., Manea, F., Nowotka, D.: Local patterns. In: Proc.
  37th {IARCS} Annual Conference on Foundations of Software Technology and
  Theoretical Computer Science, {FSTTCS} 2017. LIPIcs, vol.~93, pp. 24:1--24:14
  (2017)

\bibitem{w1hardness}
Downey, R.G., Fellows, M.R.: Parameterized Complexity. Monographs in Computer
  Science, Springer (1999)

\bibitem{FaginEtAl2015}
Fagin, R., Kimelfeld, B., Reiss, F., Vansummeren, S.: Document spanners: {A}
  formal approach to information extraction. J. {ACM}  \textbf{62}(2),
  12:1--12:51 (2015)

\bibitem{DBLP:journals/tcs/FernauMMS18}
Fernau, H., Manea, F., Mercas, R., Schmid, M.L.: Revisiting {S}hinohara's
  algorithm for computing descriptive patterns. Theor. Comput. Sci.
  \textbf{733},  44--54 (2018)

\bibitem{DBLP:journals/toct/FernauMMS20}
Fernau, H., Manea, F., Mercas, R., Schmid, M.L.: Pattern matching with
  variables: Efficient algorithms and complexity results. {ACM} Trans. Comput.
  Theory  \textbf{12}(1),  6:1--6:37 (2020)

\bibitem{FerSch2015}
Fernau, H., Schmid, M.L.: Pattern matching with variables: A multivariate
  complexity analysis. Inf. Comput.  \textbf{242},  287--305 (2015)

\bibitem{DBLP:journals/mst/FernauSV16}
Fernau, H., Schmid, M.L., Villanger, Y.: On the parameterised complexity of
  string morphism problems. Theory Comput. Syst.  \textbf{59}(1),  24--51
  (2016)

\bibitem{FleischerK18}
Fleischer, L., Kufleitner, M.: Testing {S}imon's congruence. In: Potapov, I.,
  Spirakis, P.G., Worrell, J. (eds.) {MFCS} 2018. LIPIcs, vol.~117, pp.
  62:1--62:13. Schloss Dagstuhl - Leibniz-Zentrum f{\"{u}}r Informatik (2018)

\bibitem{fleischmann2021scattered}
Fleischmann, P., Germann, S., Nowotka, D.: Scattered factor universality--the
  power of the remainder. preprint arXiv:2104.09063 (published at RuFiDim)
  (2021)

\bibitem{fleischmann2023alphabetafactorization}
Fleischmann, P., Höfer, J., Huch, A., Nowotka, D.:
  $\alpha$-$\beta$-factorization and the binary case of {S}imon's congruence
  (2023)

\bibitem{Fre2013}
Freydenberger, D.D.: Extended regular expressions: Succinctness and
  decidability. Theory of Comput. Syst.  \textbf{53},  159--193 (2013)

\bibitem{Freydenberger2019}
Freydenberger, D.D.: A logic for document spanners. Theory Comput. Syst.
  \textbf{63}(7),  1679--1754 (2019)

\bibitem{FreydenbergerGK15}
Freydenberger, D.D., Gawrychowski, P., Karhum{\"{a}}ki, J., Manea, F., Rytter,
  W.: Testing k-binomial equivalence. CoRR  \textbf{abs/1509.00622} (2015)

\bibitem{FreydenbergerHolldack2018}
Freydenberger, D.D., Holldack, M.: Document spanners: From expressive power to
  decision problems. Theory Comput. Syst.  \textbf{62}(4),  854--898 (2018)

\bibitem{FreydenbergerP21}
Freydenberger, D.D., Peterfreund, L.: The theory of concatenation over finite
  models. In: Bansal, N., Merelli, E., Worrell, J. (eds.) {ICALP} 2021,
  Proceedings. LIPIcs, vol.~198, pp. 130:1--130:17. Schloss Dagstuhl -
  Leibniz-Zentrum f{\"{u}}r Informatik (2021)

\bibitem{FreydenbergerSchmid2019}
Freydenberger, D.D., Schmid, M.L.: Deterministic regular expressions with
  back-references. J. Comput. Syst. Sci.  \textbf{105},  1--39 (2019)

\bibitem{fri:mas}
Friedl, J.E.F.: Mastering Regular Expressions. O'Reilly, Sebastopol, CA, third
  edn. (2006)

\bibitem{garelCPM}
Garel, E.: Minimal separators of two words. In: Proc. CPM 1993. Lecture Notes
  in Computer Science, vol.~684, pp. 35--53. Springer (1993)

\bibitem{gar:com}
Garey, M.R., Johnson, D.S.: Computers and Intractability: A Guide to the Theory
  of NP-Completeness. W.~H. Freeman \& Co., New York, NY, USA (1979)

\bibitem{GawrychowskiKKM21}
Gawrychowski, P., Kosche, M., Ko{\ss}, T., Manea, F., Siemer, S.: Efficiently
  testing {S}imon's congruence. In: {STACS} 2021. LIPIcs, vol.~187, pp.
  34:1--34:18. Schloss Dagstuhl - Leibniz-Zentrum f{\"{u}}r Informatik (2021)

\bibitem{GawrychowskiRSS17}
Gawrychowski, P., Lange, M., Rampersad, N., Shallit, J.O., Szykula, M.:
  Existential length universality. In: Proc. {STACS} 2020. LIPIcs, vol.~154,
  pp. 16:1--16:14 (2020)

\bibitem{mfcs2021}
Gawrychowski, P., Manea, F., Siemer, S.: Matching patterns with variables under
  {H}amming distance. In: 46th International Symposium on Mathematical
  Foundations of Computer Science, {MFCS} 2021. LIPIcs, vol.~202, pp.
  48:1--48:24 (2021)

\bibitem{spire2022}
Gawrychowski, P., Manea, F., Siemer, S.: Matching patterns with variables under
  edit distance. In: Arroyuelo, D., Poblete, B. (eds.) String Processing and
  Information Retrieval - 29th International Symposium, {SPIRE} 2022,
  Concepci{\'{o}}n, Chile, November 8-10, 2022, Proceedings. Lecture Notes in
  Computer Science, vol. 13617, pp. 275--289. Springer (2022)

\bibitem{GoecknerGHKKKS18}
Goeckner, B., Groothuis, C., Hettle, C., Kell, B., Kirkpatrick, P., Kirsch, R.,
  Solava, R.W.: Universal partial words over non-binary alphabets. Theor.
  Comput. Sci  \textbf{713} (2018)

\bibitem{Hague19}
Hague, M.: Strings at {MOSCA}. {ACM} {SIGLOG} News  \textbf{6}(4),  4--22
  (2019)

\bibitem{Hebrard91}
H{\'{e}}brard, J.: An algorithm for distinguishing efficiently bit-strings by
  their subsequences. Theoretical Computer Science  \textbf{82}(1),  35--49
  (1991)

\bibitem{Lenstra83}
Jr., H.W.L.: Integer programming with a fixed number of variables. Mathematics
  of Operations Research  \textbf{8}(4),  538--548 (1983)

\bibitem{karandikar2016height}
Karandikar, P., Schnoebelen, P.: The height of piecewise-testable languages
  with applications in logical complexity. In: CSL (2016)

\bibitem{KarhumakiSZ13}
Karhum{\"{a}}ki, J., Saarela, A., Zamboni, L.Q.: On a generalization of abelian
  equivalence and complexity of infinite words. J. Comb. Theory, Ser. {A}
  \textbf{120}(8),  2189--2206 (2013)

\bibitem{KarhumakiSZ17}
Karhum{\"{a}}ki, J., Saarela, A., Zamboni, L.Q.: Variations of the
  {M}orse-{H}edlund theorem for \emph{k}-abelian equivalence. Acta Cybern.
  \textbf{23}(1),  175--189 (2017)

\bibitem{Karp72}
Karp, R.M.: Reducibility among combinatorial problems. In: Miller, R.E.,
  Thatcher, J.W. (eds.) Proceedings of a symposium on the Complexity of
  Computer Computations. pp. 85--103. The {IBM} Research Symposia Series,
  Plenum Press, New York (1972). \doi{10.1007/978-1-4684-2001-2\_9}

\bibitem{KimHKS22}
Kim, S., Han, Y., Ko, S., Salomaa, K.: On {S}imon's congruence closure of a
  string. In: {DCFS} 2022, Proceedings. Lecture Notes in Computer Science, vol.
  13439, pp. 127--141. Springer (2022)

\bibitem{patternMatchingSimon}
Kim, S., Ko, S., Han, Y.: Simon's congruence pattern matching. In: Bae, S.W.,
  Park, H. (eds.) 33rd International Symposium on Algorithms and Computation,
  {ISAAC} 2022, December 19-21, 2022, Seoul, Korea. LIPIcs, vol.~248, pp.
  60:1--60:17. Schloss Dagstuhl - Leibniz-Zentrum f{\"{u}}r Informatik (2022)

\bibitem{SchmidICDT2022}
Kleest{-}Mei{\ss}ner, S., Sattler, R., Schmid, M.L., Schweikardt, N., Weidlich,
  M.: Discovering event queries from traces: Laying foundations for
  subsequence-queries with wildcards and gap-size constraints. In: 25th
  International Conference on Database Theory, {ICDT} 2022. LIPIcs, vol.~220,
  pp. 18:1--18:21 (2022)

\bibitem{kosche2021absent}
Kosche, M., Ko{\ss}, T., Manea, F., Siemer, S.: Absent subsequences in words.
  In: RP. pp. 115--131. Springer (2021)

\bibitem{Kosche2022SubsequenceSurvey}
Kosche, M., Ko{\ss}, T., Manea, F., Siemer, S.: Combinatorial algorithms for
  subsequence matching: A survey. In: Bordihn, H., Horv\'ath, G., Vaszil, G.
  (eds.) NCMA. vol.~367, pp. 11--27 (2022)

\bibitem{KrotzschMT17}
Kr{\"{o}}tzsch, M., Masopust, T., Thomazo, M.: Complexity of universality and
  related problems for partially ordered {NFA}s. Inf. Comput.  \textbf{255},
  177--192 (2017)

\bibitem{LejeuneRR20}
Lejeune, M., Rigo, M., Rosenfeld, M.: The binomial equivalence classes of
  finite words. Int. J. Algebra Comput.  \textbf{30}(07),  1375--1397 (2020)

\bibitem{lothaire}
Lothaire, M.: Combinatorics on Words. Cambridge University Press (1997)

\bibitem{Loth02}
Lothaire, M.: Algebraic Combinatorics on Words. Cambridge University Press
  (2002)

\bibitem{martin1934}
Martin, M.H.: A problem in arrangements. Bull. Amer. Math. Soc.
  \textbf{40}(12),  859--864 (12 1934)

\bibitem{Rampersad:2012}
Rampersad, N., Shallit, J., Xu, Z.: The computational complexity of
  universality problems for prefixes, suffixes, factors, and subwords of
  regular languages. Fundam. Inf.  \textbf{116}(1-4),  223--236 (Jan 2012)

\bibitem{ReidenbachS14}
Reidenbach, D., Schmid, M.L.: Patterns with bounded treewidth. Inf. Comput.
  \textbf{239},  87--99 (2014), \url{https://doi.org/10.1016/j.ic.2014.08.010}

\bibitem{RigoS15}
Rigo, M., Salimov, P.: Another generalization of abelian equivalence: Binomial
  complexity of infinite words. Theor. Comput. Sci.  \textbf{601},  47--57
  (2015)

\bibitem{schmid13}
Schmid, M.L.: A note on the complexity of matching patterns with variables.
  Inf. Process. Lett.  \textbf{113}(19),  729--733 (2013)

\bibitem{SchmidSchweikardt2021}
Schmid, M.L., Schweikardt, N.: A purely regular approach to non-regular core
  spanners. In: Proc. 24th International Conference on Database Theory, {ICDT}
  2021. LIPIcs, vol.~186, pp. 4:1--4:19 (2021)

\bibitem{SchmidSchweikardtPODS2022}
Schmid, M.L., Schweikardt, N.: Document spanners - {A} brief overview of
  concepts, results, and recent developments. In: {PODS} '22: International
  Conference on Management of Data. pp. 139--150. {ACM} (2022)

\bibitem{schnoebelen2019height}
Schnoebelen, P., Karandikar, P.: The height of piecewise-testable languages and
  the complexity of the logic of subwords. Logical Methods in Computer Science
  \textbf{15} (2019)

\bibitem{SchnoebelenV23}
Schnoebelen, P., Veron, J.: On arch factorization and subword universality for
  words and compressed words. In: {WORDS} 2023, Proceedings. Lecture Notes in
  Computer Science, vol. 13899, pp. 274--287 (2023)

\bibitem{shi:pol2}
Shinohara, T.: Polynomial time inference of pattern languages and its
  application. In: Proc. 7th IBM Symposium on Mathematical Foundations of
  Computer Science, MFCS. pp. 191--209 (1982)

\bibitem{shi:pat}
Shinohara, T., Arikawa, S.: Pattern inference. In: Algorithmic Learning for
  Knowledge-Based Systems, GOSLER Final Report. LNAI, vol.~961, pp. 259--291
  (1995)

\bibitem{simonPhD}
Simon, I.: Hierarchies of events with dot-depth one. Ph.D. thesis, University
  of Waterloo (1972)

\bibitem{Simon75}
Simon, I.: Piecewise testable events. In: Barkhage, H. (ed.) Automata Theory
  and Formal Languages, 2nd {GI} Conference, Kaiserslautern, May 20-23, 1975.
  Lecture Notes in Computer Science, vol.~33, pp. 214--222. Springer (1975)

\bibitem{SimonWords}
Simon, I.: Words distinguished by their subwords (extended abstract). In: Proc.
  {WORDS} 2003. TUCS General Publication, vol.~27, pp. 6--13 (2003)

\bibitem{DBLP:conf/wia/Tronicek02}
Tron{\'{\i}}cek, Z.: Common subsequence automaton. In: Proc. {CIAA} 2002
  (Revised Papers). Lecture Notes in Computer Science, vol.~2608, pp. 270--275
  (2002)

\bibitem{DBLP:journals/corr/abs-0907-0616}
Weis, P., Immerman, N.: Structure theorem and strict alternation hierarchy for
  {FO{\^{}}2} on words. Log. Methods Comput. Sci.  \textbf{5}(3) (2009)

\end{thebibliography}

\end{document}